\documentclass[12pt,technote, onecolumn, draft]{IEEEtran}
\usepackage{latexsym}
\usepackage{mathrsfs}
\usepackage{multirow}
\usepackage{amsfonts,amssymb,amsmath,amsthm,bm}
\usepackage{color}
\usepackage{cite}
\newtheorem{theorem}{Theorem}
\newtheorem{example}{Example}

\newtheorem{remark}{Remark}
\newtheorem{corollary}{Corollary}
\newtheorem{lemma}{Lemma}
\newtheorem{proposition}{Proposition}

\begin{document}

\title{On the Dual of Generalized Bent Functions$^{\dag}$}
\author{Jiaxin Wang*, Fang-Wei Fu
\IEEEcompsocitemizethanks{\IEEEcompsocthanksitem Jiaxin Wang and Fang-Wei Fu are with Chern Institute of Mathematics and LPMC, Nankai University, Tianjin 300071, P.R.China, Emails: wjiaxin@mail.nankai.edu.cn, fwfu@nankai.edu.cn.
}
\thanks{$^\dag$This research is supported by the National Key Research and Development Program of China (Grant No. 2018YFA0704703), the National Natural Science Foundation of China (Grant No. 61971243), the Natural Science Foundation of Tianjin (20JCZDJC00610), the Fundamental Research Funds for the Central Universities of China (Nankai University), and the Nankai Zhide Foundation.}
\thanks{*Corresponding author}
\thanks{Manuscript submitted  June 6, 2021}}

\maketitle

\begin{abstract}
   In this paper, we study the dual of generalized bent functions $f: V_{n}\rightarrow \mathbb{Z}_{p^k}$ where $V_{n}$ is an $n$-dimensional vector space over $\mathbb{F}_{p}$ and $p$ is an odd prime, $k$ is a positive integer. It is known that weakly regular generalized bent functions always appear in pairs since the dual of a weakly regular generalized bent function is also a weakly regular generalized bent function. The dual of non-weakly regular generalized bent functions can be generalized bent or not generalized bent. By generalizing the construction of \cite{Cesmelioglu5}, we obtain an explicit construction of generalized bent functions for which the dual can be generalized bent or not generalized bent. We show that the generalized indirect sum construction method given in \cite{Wang} can provide a secondary construction of generalized bent functions for which the dual can be generalized bent or not generalized bent. By using the knowledge on ideal decomposition in cyclotomic field, we prove that $f^{**}(x)=f(-x)$ if $f$ is a generalized bent function and its dual $f^{*}$ is also a generalized bent function. For any non-weakly regular generalized bent function $f$ which satisfies that $f(x)=f(-x)$ and its dual $f^{*}$ is generalized bent, we give a property and as a consequence, we prove that there is no self-dual generalized bent function $f: V_{n}\rightarrow \mathbb{Z}_{p^k}$ if $p\equiv 3 \ (mod \ 4)$ and $n$ is odd. For $p \equiv 1 \ (mod \ 4)$ or $p\equiv 3 \ (mod \ 4)$ and $n$ is even, we give a secondary construction of self-dual generalized bent functions. In the end, we characterize the relations between the generalized bentness of the dual of generalized bent functions and the bentness of the dual of bent functions, as well as the self-duality relations between generalized bent functions and bent functions by the decomposition of generalized bent functions.
\end{abstract}


\section{Introduction}
Throughout this paper, let $p$ be an odd prime, $\mathbb{Z}_{p^k}$ be the ring of integers modulo $p^k$, $\mathbb{F}_{p}^{n}$ be the vector space of the $n$-tuples over $\mathbb{F}_{p}$, $\mathbb{F}_{p^n}$ be the finite field with $p^n$ elements and $V_{n}$ be an $n$-dimensional vector space over $\mathbb{F}_{p}$ where $k, n$ are positive integers.

Boolean bent functions introduced by Rothaus \cite{Rothaus} play an important role in cryptography, coding theory, sequences and combinatorics. In 1985, Kumar \emph{et al}. \cite{Kumar} generalized Boolean bent functions to bent functions over finite fields of odd characteristic. Due to the importance of bent functions, they have been studied extensively. There is an exhaustive survey \cite{Carlet} and a book \cite{Mesnager1} for bent functions and generalized bent functions. Recently, the notion of Boolean generalized bent functions has been generalized to generalized bent functions from $V_{n}$ to $\mathbb{Z}_{p^k}$ \cite{Mesnager2}.

A function $f$ from $V_{n}$ to $\mathbb{Z}_{p^k}$ is called a generalized $p$-ary function, or simply $p$-ary function when $k=1$. The Walsh transform of a generalized $p$-ary function $f: V_{n} \rightarrow \mathbb{Z}_{p^k}$ is the complex valued function on $V_{n}$ defined as
\begin{equation}\label{1}
  W_{f}(a)=\sum_{x\in V_{n}}\zeta_{p^k}^{f(x)}\zeta_{p}^{-\langle a, x \rangle}, a \in V_{n}
 \end{equation}
where $\zeta_{p^k}=e^{\frac{2\pi \sqrt{-1}}{p^k}}$ is the complex primitive $p^k$-th root of unity and $\langle a, x \rangle$ denotes a (nondegenerate) inner product in $V_{n}$. And $f$ can be recovered by the inverse transform
 \begin{equation}\label{2}
  \zeta_{p^k}^{f(x)}=\frac{1}{p^n}\sum_{a\in V_{n}}W_{f}(a)\zeta_{p}^{\langle a, x\rangle}, x \in V_{n}.
 \end{equation}
The classical representations of $V_{n}$ are $\mathbb{F}_{p}^{n}$ with $\langle a, x \rangle= a \cdot x$ and $\mathbb{F}_{p^n}$ with $\langle a, x \rangle =Tr_{1}^{n}(ax)$, where $a\cdot x$ is the usual dot product over $\mathbb{F}_{p}^{n}$, $Tr_{1}^{n}(\cdot)$ is the absolute trace function. When $V_{n}=V_{n_{1}}\times \dots \times V_{n_{s}} (n=\sum_{i=1}^{s}n_{i}$), let $\langle a, b\rangle =\sum_{i=1}^{s}\langle a_{i}, b_{i}\rangle $ where $a=(a_{1}, \dots, a_{s}), b=(b_{1}, \dots, b_{s})\in V_{n}$.

The generalized $p$-ary function $f: V_{n}\rightarrow \mathbb{Z}_{p^k}$ is called a generalized $p$-ary bent function, or simply $p$-ary bent function when $k=1$ if $|W_{f}(a)|=p^{\frac{n}{2}}$ for any $a\in V_{n}$. In \cite{Mesnager2}, the authors have shown that the Walsh transform of a generalized $p$-ary bent function $f: V_{n}\rightarrow \mathbb{Z}_{p^k}$ satisfies that for any $a \in V_{n}$,
\begin{equation}\label{3}
  W_{f}(a)=\left\{\begin{split}
                     \pm p^{\frac{n}{2}}\zeta_{p^k}^{f^{*}(a)} & \ \ if \ p \equiv 1 \ (mod \ 4) \ or  \ n \ is \ even,\\
                     \pm \sqrt{-1} p^{\frac{n}{2}}\zeta_{p^k}^{f^{*}(a)} & \ \ if \ p \equiv 3 \ (mod \ 4) \ and  \ n \ is \ odd
                  \end{split}
  \right.
 \end{equation}
where $f^{*}$ is a function from $V_{n}$ to $\mathbb{Z}_{p^k}$ and is called the dual of $f$. We call $f$ self-dual if $f^{*}=f$. In the sequel, if $f^{*}$ is also a generalized bent function, let $f^{**}$ denote the dual of $f^{*}$.

If $f: V_{n}\rightarrow \mathbb{Z}_{p^k}$ is a generalized bent function with dual $f^{*}$, define
\begin{equation}\label{4}
  \mu_{f}(a)=p^{-\frac{n}{2}}\zeta_{p^k}^{-f^{*}(a)}W_{f}(a), a \in V_{n}
\end{equation}
and
\begin{equation}\label{5}
  \varepsilon_{f}(a)=\xi^{-1}\mu_{f}(a), a \in V_{n}
\end{equation}
where $\xi=1$ if $p\equiv 1 \ (mod\ 4)$ or $n$ is even and $\xi=\sqrt{-1}$ if $p\equiv 3 \ (mod\ 4)$ and $n$ is odd.
By (3), it is easy to see that $\varepsilon_{f}$ is a function from $V_{n}$ to $\{\pm 1\}$. If $\mu_{f}$ is a constant function, then $f$ is called weakly regular, otherwise $f$ is called non-weakly regular. In particular, if $\mu_{f}(a)=1, a \in V_{n}$, $f$ is called regular.

There are some works on the dual of $p$-ary bent functions \cite{Cesmelioglu2}, \cite{Cesmelioglu4}, \cite{Cesmelioglu5}, \cite{Ozbudak}. It is known that weakly regular generalized bent functions always appear in pairs since the dual of a weakly regular generalized bent function is also a weakly regular generalized bent function \cite{Mesnager2}. For non-weakly regular bent functions, the dual can be bent or not bent. In \cite{Cesmelioglu4}, \c{C}e\c{s}melio\u{g}lu \emph{et al.} analysed the first known construction of non-weakly regular bent functions \cite{Cesmelioglu1} and showed that this construction yields bent functions for which the dual is bent if using some proper weakly regular plateaued functions or bent functions whose dual is also bent as building blocks. In \cite{Cesmelioglu5}, \c{C}e\c{s}melio\u{g}lu \emph{et al.} provided the first explicit construction of non-weakly regular bent functions for which the dual is not bent. As a consequence, the dual of non-weakly regular generalized bent functions can be generalized bent or not generalized bent. For instance, if $f: V_{n}\rightarrow \mathbb{F}_{p}$ is a non-weakly regular bent function whose dual is (not) bent, then obviously $p^{k-1}f: V_{n}\rightarrow \mathbb{Z}_{p^k}$ is a non-weakly regular generalized bent function whose dual is (not) generalized bent. In \cite{Cesmelioglu4}, the authors also analysed the self-duality of $p$-ary bent functions from $V_{n}$ to $\mathbb{F}_{p}$. They characterized quadratic self-dual bent functions. For $p\equiv 1 \ (mod \ 4)$, they proposed a secondary construction of self-dual bent functions, which can be used to construct non-quadratic self-dual bent functions by using (non)-quadratic self-dual bent functions as building blocks. However, for $p\equiv 3 \ (mod \ 4)$ and $n$ is even, they only illustrated that there exist ternary non-quadratic self-dual bent functions by considering special ternary bent functions. For $p \equiv 3 \ (mod \ 4)$ and $n$ is odd, they showed that there is no weakly regular self-dual bent function and they pointed out that there exist non-weakly regular self-dual bent functions. Indeed, they pointed out that the non-weakly regular bent function $g_{3}(x)=Tr_{1}^{3}(x^{22}+x^{8})$ from $\mathbb{F}_{3^3}$ to $\mathbb{F}_{3}$ is self-dual. But in fact, $g_{3}$ is not self-dual. In this paper, we will prove that there is no self-dual generalized bent function $f: V_{n}\rightarrow \mathbb{Z}_{p^k}$ if $p\equiv 3 \ (mod \ 4)$ and $n$ is odd. For weakly regular bent functions $f$ with dual $f^{*}$, it is known that $f^{**}(x)=f(-x)$ holds where $f^{**}$ denotes the dual of $f^{*}$. Recently, for non-weakly regular bent functions $f$ whose dual $f^{*}$ is also bent, \"{O}zbudak and Pelen \cite{Ozbudak} showed that $f^{**}(x)=f(-x)$ also holds by studying the value distributions of the dual of non-weakly regular bent functions.

In this paper, we study the dual of generalized bent functions from $V_{n}$ to $\mathbb{Z}_{p^k}$. By generalizing the construction of Corollary 2 of \cite{Cesmelioglu5}, we obtain an explicit construction of generalized bent functions for which the dual can be generalized bent or not generalized bent. Note that in \cite{Cesmelioglu5}, a sufficient condition is given to construct non-weakly regular bent functions whose dual is not bent, however, a sufficient condition for the dual of the constructed non-weakly regular bent function to be bent is not given. We will give a sufficient condition to illustrate that the construction of Corollary 2 of \cite{Cesmelioglu5} can also be used to construct non-weakly regular bent functions whose dual is bent. We will show that the generalized indirect sum construction method given in \cite{Wang} can provide a secondary construction of generalized bent functions for which the dual can be generalized bent or not generalized bent. For generalized bent functions $f$ whose dual $f^{*}$ is also generalized bent, different from the proof method in \cite{Ozbudak}, we prove $f^{**}(x)=f(-x)$ by using the knowledge on ideal decomposition in cyclotomic field, which seems more concise. For any non-weakly regular generalized bent function $f$ which satisfies that $f(x)=f(-x)$ and the dual $f^{*}$ is generalized bent, we give a property and as a consequence, we prove that there is no self-dual generalized bent function $f: V_{n}\rightarrow \mathbb{Z}_{p^k}$ if $p\equiv 3 \ (mod \ 4)$ and $n$ is odd. For $p \equiv 1 \ (mod \ 4)$ or $p\equiv 3 \ (mod \ 4)$ and $n$ is even, we give a secondary construction of self-dual generalized bent functions. In the end, by the decomposition of generalized bent functions, we characterize the relations between the generalized bentness of the dual of generalized bent functions and the bentness of the dual of bent functions, as well as the self-duality relations between generalized bent functions and bent functions.

The rest of the paper is organized as follows. In Section II, we present some needed results on generalized bent functions and number fields. In Section III, we give constructions of generalized bent functions whose dual can be generalized bent or not generalized bent (Theorems 1 and 2). In Section IV, for generalized bent functions $f$ whose dual $f^{*}$ is also generalized bent, by using the knowledge on ideal decomposition in cyclotomic field, we prove that $f^{**}(x)=f(-x)$ holds (Theorem 3). In Section V, we prove that there is no self-dual generalized bent function $f: V_{n}\rightarrow \mathbb{Z}_{p^k}$ if $p\equiv 3 \ (mod \ 4)$ and $n$ is odd (Theorem 4). In Section VI, we give a secondary construction of self-dual generalized bent functions $f: V_{n}\rightarrow \mathbb{Z}_{p^k}$ if $p\equiv 1 \ (mod \ 4)$ or $n$ is even (Theorem 5). In Section VII, we characterize the relations between the generalized bentness of the dual of generalized bent functions and the bentness of the dual of bent functions, as well as the self-duality relations between generalized bent functions and bent functions (Theorem 6). In Section VIII, we make a conclusion.
\section{Preliminaries}
\label{sec:1}
In this section, we present some needed results on generalized bent functions and number fields.
\subsection{Some Results on Generalized Bent Functions}
Let $f: V_{n}\rightarrow \mathbb{Z}_{p^k}$ be a generalized bent function, then $W_{f}(a)=\mu_{f}(a)p^{\frac{n}{2}}\zeta_{p^k}^{f^{*}(a)}$ for any $a\in V_{n}$ where $\mu_{f}(a)=\xi \varepsilon_{f}(a)$, $\xi=1$ if $p\equiv 1 \ (mod\ 4)$ or $n$ is even and $\xi=\sqrt{-1}$ if $p\equiv 3 \ (mod\ 4)$ and $n$ is odd, $\varepsilon_{f}: V_{n} \rightarrow \{\pm 1\}$, $f^{*}$ is the dual of $f$. By the inverse transform (2), we have
\begin{equation*}
\begin{split}
  \zeta_{p^k}^{f(x)}&=\frac{1}{p^n}\sum_{a \in V_{n}}W_{f}(a)\zeta_{p}^{\langle a, x\rangle} \\
                    &=\frac{1}{p^n}\sum_{a \in V_{n}}\xi\varepsilon_{f}(a)p^{\frac{n}{2}}\zeta_{p^k}^{f^{*}(a)}\zeta_{p}^{\langle a, x\rangle}\\
                    &=\xi p^{-\frac{n}{2}}\sum_{a \in V_{n}}\varepsilon_{f}(a)\zeta_{p^k}^{f^{*}(a)}\zeta_{p}^{\langle a, x\rangle}.
\end{split}
\end{equation*}
Hence for any $x \in V_{n}$,
\begin{equation}\label{6}
  \xi\sum_{a \in V_{n}}\varepsilon_{f}(a)\zeta_{p^k}^{f^{*}(a)}\zeta_{p}^{\langle a, x\rangle}=p^{\frac{n}{2}}\zeta_{p^k}^{f(x)}.
\end{equation}

In this paper, let $\eta$ be the multiplicative quadratic character of $\mathbb{F}_{p^n}$, that is,
\begin{equation*}
  \eta(x)=\left\{\begin{split}
                   1 & \ \ if \ x \in \mathbb{F}_{p^n}^{*} \ is \ a \ square \ element, \\
                   -1 & \ \ if \ x \in \mathbb{F}_{p^n}^{*} \ is \ a \ non\raisebox{0mm}{-}square \ element.
                 \end{split}\right.
\end{equation*} Let $\alpha \in \mathbb{F}_{p^n}^{*}$. Then the function $f: \mathbb{F}_{p^n}\rightarrow \mathbb{F}_{p}$ defined by $f(x)=Tr_{1}^{n}(\alpha x^{2})$ is a weakly regular bent function and $W_{f}(a)=(-1)^{n-1}\epsilon^{n}\eta(\alpha)p^{\frac{n}{2}}\zeta_{p}^{Tr_{1}^{n}(-\frac{a^2}{4\alpha})}$ for any $a \in \mathbb{F}_{p^n}$ where $\epsilon=1$ if $p\equiv 1 \ (mod \ 4)$ and $\epsilon=\sqrt{-1}$ if $p\equiv 3 \ (mod \ 4)$ (see \cite{Helleseth}).

In \cite{Mesnager2}, the authors characterized the relations between generalized bent functions and bent functions by the decomposition of generalized bent functions. The following lemma follows from Theorem 16 and its proof of \cite{Mesnager2}.
\begin{lemma} [\cite{Mesnager2}]
Let $k\geq 2$ be an integer. Let $f: V_{n}\rightarrow \mathbb{Z}_{p^k}$ with the decomposition $f(x)=\sum_{i=0}^{k-1}f_{i}(x)p^{k-1-i}=f_{0}(x)p^{k-1}+\bar{f}_{1}(x)$ where $f_{i}$ is a function from $V_{n}$ to $\mathbb{F}_{p}$ for any $0\leq i \leq k-1$ and $\bar{f}_{1}(x)=\sum_{i=1}^{k-1}f_{i}(x)p^{k-1-i}$ is a function from $V_{n}$ to $\mathbb{Z}_{p^{k-1}}$. Then $f$ is a generalized bent function if and only if $g_{f, F}\triangleq f_{0}+F(f_{1}, \dots, f_{k-1})$ is a bent function for any function $F: \mathbb{F}_{p}^{k-1}\rightarrow \mathbb{F}_{p}$. Furthermore, if $f$ is a generalized bent function, then $f^{*}(x)=f_{0}^{*}(x)p^{k-1}+\lambda(x)$ and $g_{f, F}^{*}(x)=f_{0}^{*}(x)+F(\lambda_{1}(x), \dots, \lambda_{k-1}(x))$ where $\lambda(x)=\sum_{i=1}^{k-1}\lambda_{i}(x)p^{k-1-i}, \lambda_{i}: V_{n}\rightarrow \mathbb{F}_{p}$ and $\lambda: V_{n}\rightarrow \mathbb{Z}_{p^{k-1}}$ satisfies that for any $a \in V_{n}$, $\sum_{x \in V_{n}: \bar{f}_{1}(x)=\lambda(a)}\zeta_{p}^{f_{0}(x)-\langle a, x\rangle}$ $=W_{f_{0}}(a)$ and $\sum_{x \in V_{n}:\bar{f}_{1}(x)=v}\zeta_{p}^{f_{0}(x)-\langle a, x\rangle}=0$ for any $v \in \mathbb{Z}_{p^{k-1}}$ with $v \neq \lambda(a)$.
\end{lemma}
\subsection{Some Results on Number Fields}
In this subsection, we give some results on number fields. For more results on number fields, we refer to \cite{Feng}.

Suppose $L/K$ is a number field extension. Let $\mathcal{O}_{L}$ and $\mathcal{O}_{K}$ be the ring of integers in $L$ and $K$ respectively.
Any nonzero ideal $I$ of $\mathcal{O}_{L}$ can be uniquely (up to the order) expressed as
\begin{equation*}
  I=\mathfrak{B}_{1}^{e_{1}}\dots \mathfrak{B}_{g}^{e_{g}}
\end{equation*}
where $\mathfrak{B}_{1}, \dots, \mathfrak{B}_{g}$ are distinct (nonzero) prime ideals of $\mathcal{O}_{L}$ and $e_{i}\geq 1 \ (1\leq i \leq g)$. And when $I=\mathfrak{p}\mathcal{O}_{L}$ where $\mathfrak{p}$ is a nonzero prime ideal of $\mathcal{O}_{K}$, $\mathfrak{p}=\mathfrak{B}_{i}\cap \mathcal{O}_{K}$ for any $1\leq i \leq g$.

In this paper, we consider cyclotomic field $L=\mathbb{Q}(\zeta_{p^k})$. Then $\mathcal{O}_{L}=\mathbb{Z}[\zeta_{p^k}]$ and $\{\zeta_{p}^{i}\zeta_{p^k}^{j}: 0\leq i \leq p-2, 0\leq j \leq p^{k-1}-1\}$ is an integral basis of $\mathcal{O}_{L}$. For prime ideal $p\mathbb{Z}$ of $\mathbb{Z}$, $(p\mathbb{Z})\mathcal{O}_{L}=p\mathcal{O}_{L}$ and
\begin{equation*}
   p\mathcal{O}_{L}=((1-\zeta_{p^k})\mathcal{O}_{L})^{\varphi(p^k)}
\end{equation*}
where $(1-\zeta_{p^k})\mathcal{O}_{L}$ is a prime ideal of $\mathcal{O}_{L}$ and $\varphi$ is the Euler function.

For any integer $j\in \mathbb{Z}$, let $v_{p}(j)\in \mathbb{Z}$ denote the $p$-adic valuation of $j$, that is, $p^{v_{p}(j)}\mid j$ and $p^{v_{p}(j)+1}\nmid j$.
\section{Constructions of generalized bent functions whose dual is (not) generalized bent}
In this section, we give constructions of generalized bent functions whose dual can be generalized bent or not generalized bent.

The following theorem gives an explicit construction of generalized bent functions for which the dual can be generalized bent or not generalized bent, which is a generalization of Corollary 2 of \cite{Cesmelioglu5}. Note that in Corollary 2 of \cite{Cesmelioglu5}, a sufficient condition is given to construct non-weakly regular bent functions whose dual is not bent, however, a sufficient condition for the dual of the constructed non-weakly regular bent function to be bent is not given. In the following theorem, we also give a sufficient condition to illustrate that the construction of \cite{Cesmelioglu5} can also be used to construct non-weakly regular bent functions whose dual is bent.
\begin{theorem}
Let $m\geq 2$ be an integer. Let $\alpha, \beta \in \mathbb{F}_{p^m}^{*}$ satisfy $1+i \alpha +j \beta \neq 0$ for any $i, j \in \mathbb{F}_{p}$. Let $F(x, y_{1}, y_{2})=p^{k-1}(Tr_{1}^{m}(x^{2})+(y_{1}+Tr_{1}^{m}(\alpha x^{2}))(y_{2}+Tr_{1}^{m}(\beta x^{2})))+g(y_{2}+Tr_{1}^{m}(\beta x^{2})), (x, y_{1}, y_{2})\in \mathbb{F}_{p^m} \times \mathbb{F}_{p} \times \mathbb{F}_{p}$, where $g$ is an arbitrary function from $\mathbb{F}_{p}$ to $\mathbb{Z}_{p^{k}}$. Then

1) $F$ is a generalized bent function and it is weakly regular if and only if $\eta(1+i \alpha +j \beta)=1$  for any $i, j \in \mathbb{F}_{p}$.

2) The dual $F^{*}(x, y_{1}, y_{2})=p^{k-1}(Tr_{1}^{m}(-\frac{x^{2}}{4(1+y_{1}\alpha+y_{2}\beta)})-y_{1}y_{2})+g(y_{1})$ and $F^{*}$ is a generalized bent function if and only if
      \begin{equation}\label{7}
        |\sum_{y_{1}, y_{2} \in \mathbb{F}_{p}}\eta(1+y_{1} \alpha +y_{2} \beta)\zeta_{p^k}^{g(y_{1})}\zeta_{p}^{-y_{1}y_{2}+b_{1}y_{1}+b_{2}y_{2}}|=p \ \ for \ any \  b_{1}, b_{2} \in \mathbb{F}_{p}.
      \end{equation}
  In particular, if $\eta(1+i \alpha +j \beta)=\eta(1+i\alpha)$ for any $i, j \in \mathbb{F}_{p}$ and there exist $i_{1}, i_{2}\in \mathbb{F}_{p}$ such that $\eta(1+i_{1}\alpha)\neq \eta(1+i_{2}\alpha)$, then $F^{*}$ is a non-weakly regular generalized bent function.
\end{theorem}
\begin{proof}
1) For any $(a, b_{1}, b_{2})\in \mathbb{F}_{p^m}\times \mathbb{F}_{p} \times \mathbb{F}_{p}$, we have
\begin{equation*}
  \begin{split}
    &W_{F}(a, b_{1}, b_{2})\\
    & =\sum_{x \in \mathbb{F}_{p^m}, y_{1}, y_{2} \in \mathbb{F}_{p}}\zeta_{p^k}^{F(x, y_{1}, y_{2})}\zeta_{p}^{-Tr_{1}^{m}(ax)-b_{1}y_{1}-b_{2}y_{2}}\\
    & =\sum_{x \in \mathbb{F}_{p^m}}\zeta_{p}^{Tr_{1}^{m}(x^{2})-Tr_{1}^{m}(ax)}\sum_{y_{1}, y_{2}\in \mathbb{F}_{p}}\zeta_{p^k}^{g(y_{2}+Tr_{1}^{m}(\beta x^{2}))}\zeta_{p}^{(y_{1}+Tr_{1}^{m}(\alpha x^{2}))(y_{2}+Tr_{1}^{m}(\beta x^{2}))-b_{1}y_{1}-b_{2}y_{2}} \\
    & =\sum_{x \in \mathbb{F}_{p^m}}\zeta_{p}^{Tr_{1}^{m}(\gamma_{b}x^{2}-ax)}\sum_{z_{1}, z_{2}\in \mathbb{F}_{p}}\zeta_{p}^{z_{1}z_{2}-b_{1}z_{1}-b_{2}z_{2}}\zeta_{p^k}^{g(z_{2})}
    \end{split}
\end{equation*}
\ \ \ \ \ \ \ \ \ $=W_{Tr_{1}^{m}(\gamma_{b}x^{2})}(a)W_{p^{k-1}z_{1}z_{2}+g(z_{2})}(b_{1}, b_{2})$
\\where $\gamma_{b}=1+b_{1}\alpha +b_{2}\beta$ and in the third equality we use the change of variables $z_{1}=y_{1}+Tr_{1}^{m}(\alpha x^{2})$, $z_{2}=y_{2}+Tr_{1}^{m}(\beta x^{2})$. Since $\gamma_{b}\neq 0$, then $Tr_{1}^{m}(\gamma_{b}x^{2})$ is a bent function and $W_{Tr_{1}^{m}(\gamma_{b}x^{2})}(a)=(-1)^{m-1}\epsilon^{m}\eta(\gamma_{b})p^{\frac{m}{2}}\zeta_{p}^{Tr_{1}^{m}(-\frac{a^{2}}{4\gamma_{b}})}$ where $\epsilon=1$ if $p\equiv 1\ (mod \ 4)$ and $\epsilon=\sqrt{-1}$ if $p\equiv 3 \ (mod \ 4)$. Since $p^{k-1}z_{1}z_{2}+g(z_{2})$ is a generalized Maiorana-McFarland bent function, then $W_{p^{k-1}z_{1}z_{2}+g(z_{2})}(b_{1}, b_{2})=p \zeta_{p^k}^{-p^{k-1}b_{1}b_{2}+g(b_{1})}$. Hence $F$ is generalized bent and
\begin{equation}\label{8}
  W_{F}(a, b_{1}, b_{2})=\mu_{F}(a, b_{1}, b_{2})p^{\frac{m+2}{2}}\zeta_{p^k}^{F^{*}(a, b_{1}, b_{2})}
\end{equation}
where $\mu_{F}(a, b_{1}, b_{2})=(-1)^{m-1}\epsilon^{m}\eta(\gamma_{b})$ and its dual $F^{*}(a, b_{1}, b_{2})=p^{k-1}(Tr_{1}^{m}(-\frac{a^{2}}{4\gamma_{b}})-b_{1}b_{2})+g(b_{1})$.

From (8), it is obviously that $F$ is weakly regular if and only if $\eta(\gamma_{b})=\eta(1+b_{1}\alpha+b_{2} \beta), b_{1}, b_{2} \in \mathbb{F}_{p}$ are the same. By $\eta(1)=1$, we have $F$ is weakly regular if and only if $\eta(\gamma_{b})=1$ for any $b_{1}, b_{2} \in \mathbb{F}_{p}$.

2) By $F^{*}(x, y_{1}, y_{2})=p^{k-1}(Tr_{1}^{m}(-\frac{x^{2}}{4\gamma_{y}})-y_{1}y_{2})+g(y_{1})$ where $\gamma_{y}=1+y_{1}\alpha+y_{2}\beta$, for any $(a, b_{1}, b_{2})\in \mathbb{F}_{p^m} \times \mathbb{F}_{p} \times \mathbb{F}_{p}$ we have
\begin{equation*}
  \begin{split}
    & W_{F^{*}}(a, b_{1}, b_{2})\\
    & =\sum_{y_{1}, y_{2}\in \mathbb{F}_{p}}\zeta_{p}^{-y_{1}y_{2}-b_{1}y_{1}-b_{2}y_{2}}\zeta_{p^k}^{g(y_{1})}\sum_{x \in \mathbb{F}_{p^m}}\zeta_{p}^{Tr_{1}^{m}(-\frac{x^{2}}{4\gamma_{y}}-ax)}\\
    & =\sum_{y_{1}, y_{2}\in \mathbb{F}_{p}}\zeta_{p}^{-y_{1}y_{2}-b_{1}y_{1}-b_{2}y_{2}}\zeta_{p^k}^{g(y_{1})}(-1)^{m-1}\epsilon^{m}p^{\frac{m}{2}}\eta(-\frac{1}{4\gamma_{y}}) \zeta_{p}^{Tr_{1}^{m}(\gamma_{y}a^{2})}\\
    & = \lambda\sum_{y_{1}, y_{2}\in \mathbb{F}_{p}}\eta(\gamma_{y})\zeta_{p}^{-y_{1}y_{2}-(b_{1}-Tr_{1}^{m}(a^{2}\alpha))y_{1}-(b_{2}-Tr_{1}^{m}(a^{2}\beta))y_{2}}
    \zeta_{p^k}^{g(y_{1})}
  \end{split}
\end{equation*}
where $\lambda=(-1)^{m-1}\epsilon^{m}\eta(-1)p^{\frac{m}{2}}\zeta_{p}^{Tr_{1}^{m}(a^{2})}$,
hence it is easy to see that $F^{*}$ is generalized bent if and only if (7) holds. In particular, if $\eta(1+i \alpha +j \beta)=\eta(1+i\alpha)$ for any $i, j \in \mathbb{F}_{p}$, then for any $b_{1}, b_{2} \in \mathbb{F}_{p}$,
\begin{equation*}
  \begin{split}
     &\sum_{y_{1}, y_{2}\in \mathbb{F}_{p}}\eta(1+y_{1}\alpha+y_{2}\beta)\zeta_{p}^{-y_{1}y_{2}+b_{1}y_{1}+b_{2}y_{2}}\zeta_{p^k}^{g(y_{1})} \\
     &=\sum_{y_{1} \in \mathbb{F}_{p}}\eta(1+y_{1}\alpha)\zeta_{p}^{b_{1}y_{1}}\zeta_{p^k}^{g(y_{1})}\sum_{y_{2}\in \mathbb{F}_{p}}\zeta_{p}^{(b_{2}-y_{1})y_{2}}\\
     &=\eta(1+b_{2}\alpha) p \zeta_{p}^{b_{1}b_{2}}\zeta_{p^k}^{g(b_{2})},
  \end{split}
  \end{equation*}
thus (7) holds and $F^{*}$ is generalized bent. Furthermore, if there exist $i_{1}, i_{2}\in \mathbb{F}_{p}$ such that $\eta(1+i_{1}\alpha)\neq \eta(1+i_{2}\alpha)$, it is easy to see that $F^{*}$ is non-weakly regular.
\end{proof}

\begin{remark}
When $k=1$, $g=0$ and $1, \alpha , \beta\in \mathbb{F}_{p^m}$ are linearly independent over $\mathbb{F}_{p}$, the above construction reduces to Corollary 2 of \cite{Cesmelioglu5}, which is the first explicit construction of bent functions whose dual is not bent. As for general $\alpha, \beta\in \mathbb{F}_{p^m}^{*}$ with $1+i \alpha +j \beta \neq 0, i, j \in \mathbb{F}_{p}$, the condition (7) does not seem to hold in general, one can easily obtain non-weakly regular generalized bent functions whose dual is not generalized bent. But we emphasize that if $\alpha, \beta \in \mathbb{F}_{p^m}^{*}$ with $1+i \alpha +j \beta \neq 0, i, j \in \mathbb{F}_{p}$ satisfy that $\eta(1+i \alpha+j \beta)=\eta(1+i \alpha)$ for any $i, j \in \mathbb{F}_{p}$ and $\eta (1+i\alpha), i\in \mathbb{F}_{p}$ are not all the same, then the function constructed by the above construction is a non-weakly regular generalized bent function whose dual is also generalized bent.
\end{remark}

We give two examples of non-weakly regular generalized bent functions by using Theorem 1 for which the dual of the first example is not generalized bent and the dual of the second example is generalized bent.

\begin{example}
Let $p=5$, $m=2$, $k=3$. Let $z$ be the primitive element of $\mathbb{F}_{5^2}$ with $z^{2}+4z+2=0$. Let $\alpha=\beta=z$, $g: \mathbb{F}_{5}\rightarrow \mathbb{Z}_{5^3}$ be defined as $g(x)=x^{3}$. Then one can verify that the function $F$ constructed by Theorem 1 is a non-weakly regular generalized bent function and its dual $F^{*}$ is not generalized bent.
\end{example}

\begin{example}
Let $p=3$, $m=5$, $k=2$. Let $z$ be the primitive element of $\mathbb{F}_{3^5}$ with $z^{5}+2z+1=0$. Let $\alpha=z^{10}, \beta=z^{47}$, $g: \mathbb{F}_{3}\rightarrow \mathbb{Z}_{3^2}$ be defined as $g(x)=x$. Then one can verify that $\eta(1+j\beta)=1, \eta(1+\alpha+j\beta)=\eta(1+2\alpha+j\beta)=-1$ for any $j \in \mathbb{F}_{3}$, hence the function $F$ constructed by Theorem 1 is a non-weakly regular generalized bent function and its dual $F^{*}$ is generalized bent.
\end{example}

Now we show that the generalized indirect sum construction method given in \cite{Wang} can provide a secondary construction of generalized bent functions for which the dual can be generalized bent or not generalized bent. For the sake of completeness, we give the proof of the following lemma, which can be obtained by Theorem 5 of \cite{Wang}.
\begin{lemma}[\cite{Wang}]
Let $f_{i} (i \in \mathbb{F}_{p}^{t}) : V_{r}\rightarrow \mathbb{Z}_{p^k}$ be generalized bent functions. Let $g_{s} (0\leq s\leq t) : V_{m}\rightarrow \mathbb{F}_{p}$ be bent functions which satisfy that for any $j=(j_{1}, \dots, j_{t})\in \mathbb{F}_{p}^{t}$, $G_{j}\triangleq (1-j_{1}-\dots-j_{t})g_{0}+j_{1}g_{1}+\dots+j_{t}g_{t}$ is a bent function and $G_{j}^{*}=(1-j_{1}-\dots-j_{t})g_{0}^{*}+j_{1}g_{1}^{*}+\dots+j_{t}g_{t}^{*}$ and $\mu_{G_{j}}(y)=u, y \in V_{m}$ where $\mu_{G_{j}}$ is defined by (4) and $u$ is a constant independent of $j$. Let $g$ be an arbitrary function from $\mathbb{F}_{p}^{t}$ to $\mathbb{Z}_{p^k}$. Then $F(x,y)=f_{(g_{0}(y)-g_{1}(y), \dots, g_{0}(y)-g_{t}(y))}(x)+p^{k-1}g_{0}(y)+g(g_{0}(y)-g_{1}(y), \dots, g_{0}(y)-g_{t}(y)), (x, y) \in V_{r}\times V_{m}$ is a generalized bent function.
\end{lemma}
\begin{proof}
For any $(a, b) \in V_{r}\times V_{m}$, we have
\begin{equation*}
  \begin{split}
      &W_{F}(a, b)\\
      &=\sum_{x \in V_{r}, y \in V_{m}}\zeta_{p^k}^{f_{(g_{0}(y)-g_{1}(y), \dots, g_{0}(y)-g_{t}(y))}(x)+g(g_{0}(y)-g_{1}(y), \dots, g_{0}(y)-g_{t}(y))}\zeta_{p}^{g_{0}(y)-\langle a, x\rangle -\langle b, y \rangle}\\
      & =\sum_{i_{1}, \dots, i_{t}\in \mathbb{F}_{p}}\sum_{y: g_{0}(y)-g_{j}(y)=i_{j}, 1\leq j\leq t}\sum_{x \in V_{r}}\zeta_{p^k}^{f_{(i_{1}, \dots, i_{t})}(x)+g(i_{1}, \dots, i_{t})}\zeta_{p}^{g_{0}(y)-\langle a, x \rangle-\langle b, y\rangle}\\
      & =p^{-t}\sum_{i_{1}, \dots, i_{t} \in \mathbb{F}_{p}}\zeta_{p^k}^{g(i_{1}, \dots, i_{t})}W_{f_{(i_{1}, \dots, i_{t})}}(a)\sum_{y \in V_{m}}\zeta_{p}^{g_{0}(y)-\langle b, y\rangle}\sum_{j_{1}\in \mathbb{F}_{p}}\zeta_{p}^{(i_{1}-(g_{0}-g_{1})(y))j_{1}}\\
      & \ \ \ \ \ \ \ \dots \sum_{j_{t}\in \mathbb{F}_{p}}\zeta_{p}^{(i_{t}-(g_{0}-g_{t})(y))j_{t}}\\
      & =p^{-t}\sum_{i_{1}, \dots, i_{t} \in \mathbb{F}_{p}}\zeta_{p^k}^{g(i_{1}, \dots, i_{t})}W_{f_{(i_{1}, \dots, i_{t})}}(a)\sum_{j_{1}, \dots, j_{t}\in \mathbb{F}_{p}}\zeta_{p}^{i_{1}j_{1}+\dots+i_{t}j_{t}}W_{G_{(j_{1}, \dots, j_{t})}}(b).
 \end{split}
 \end{equation*}

As for any $j_{1}, \dots, j_{t}\in \mathbb{F}_{p}$, $G_{(j_{1}, \dots, j_{t})}\triangleq (1-j_{1}-\dots-j_{t})g_{0}+j_{1}g_{1}+\dots+j_{t}g_{t}$ is a bent function and $G_{(j_{1}, \dots, j_{t})}^{*}=(1-j_{1}-\dots-j_{t})g_{0}^{*}+j_{1}g_{1}^{*}+\dots+j_{t}g_{t}^{*}$ and $\mu_{G_{(j_{1}, \dots, j_{t})}}(y)=u, y \in V_{m}$ where $u$ is a constant independent of $(j_{1}, \dots, j_{t})$, we have
 \begin{equation}\label{9}
     \begin{split}
      &W_{F}(a, b)\\
      & =u p^{\frac{m}{2}}p^{-t}\sum_{i_{1}, \dots, i_{t} \in \mathbb{F}_{p}}\zeta_{p^k}^{g(i_{1}, \dots, i_{t})}W_{f_{(i_{1}, \dots, i_{t})}}(a)\sum_{j_{1}, \dots, j_{t}\in \mathbb{F}_{p}}\zeta_{p}^{i_{1}j_{1}+\dots+i_{t}j_{t}+(1-j_{1}-\dots-j_{t})g_{0}^{*}(b)+j_{1}g_{1}^{*}(b)+\dots+j_{t}g_{t}^{*}(b)}\\
      & =u p^{\frac{m}{2}}p^{-t}\zeta_{p}^{g_{0}^{*}(b)}\sum_{i_{1}, \dots, i_{t} \in \mathbb{F}_{p}}\zeta_{p^k}^{g(i_{1}, \dots, i_{t})}W_{f_{(i_{1}, \dots, i_{t})}}(a)\sum_{j_{1}\in \mathbb{F}_{p}}\zeta_{p}^{(g_{1}^{*}(b)-g_{0}^{*}(b)+i_{1})j_{1}} \dots \sum_{j_{t}\in \mathbb{F}_{p}}\zeta_{p}^{(g_{t}^{*}(b)-g_{0}^{*}(b)+i_{t})j_{t}}\\
      & =u p^{\frac{m}{2}}\zeta_{p}^{g_{0}^{*}(b)}\zeta_{p^k}^{g(g_{0}^{*}(b)-g_{1}^{*}(b), \dots, g_{0}^{*}(b)-g_{t}^{*}(b))}W_{f_{(g_{0}^{*}(b)-g_{1}^{*}(b), \dots, g_{0}^{*}(b)-g_{t}^{*}(b))}}(a).
  \end{split}
\end{equation}
Hence, by (9), it is easy to see that $F: V_{r}\times V_{m}\rightarrow \mathbb{Z}_{p^k}$ is a generalized bent function if $f_{i}, i\in \mathbb{F}_{p}^{t}$ are generalized bent functions from $V_{r}$ to $\mathbb{Z}_{p^k}$.
\end{proof}

Based on Lemma 2, we give the following theorem.
\begin{theorem}
With the same notation as Lemma 2. The dual of the generalized bent function constructed by Lemma 2 is a generalized bent function if and only if for any $y \in V_{m}$, the dual of $f_{(g_{0}(y)-g_{1}(y), \dots, g_{0}(y)-g_{t}(y))}$ is a generalized bent function.
\end{theorem}
\begin{proof}
By (9), we have that the dual of the generalized bent function $F$ constructed by Lemma 2 is $F^{*}(x, y)=f^{*}_{(g_{0}^{*}(y)-g_{1}^{*}(y), \dots, g_{0}^{*}(y)-g_{t}^{*}(y))}(x)+p^{k-1}g_{0}^{*}(y)+g(g_{0}^{*}(y)-g_{1}^{*}(y), \dots, g_{0}^{*}(y)-g_{t}^{*}(y))$. Note that for a weakly regular bent function $h$ with dual $h^{*}$, $h^{*}$ is also a weakly regular bent function and $h^{**}(x)=h(-x), \mu_{h^{*}}=\mu_{h}^{-1}$. Since for any $j=(j_{1}, \dots, j_{t})\in \mathbb{F}_{p}^{t}$, $G_{j}=(1-j_{1}-\dots-j_{t})g_{0}+j_{1}g_{1}+\dots+j_{t}g_{t}$ is a bent function and $G_{j}^{*}=(1-j_{1}-\dots-j_{t})g_{0}^{*}+j_{1}g_{1}^{*}+\dots+j_{t}g_{t}^{*}$ and $\mu_{G_{j}}(y)=u, y \in V_{m}$ where $u$ is a constant independent of $j$, we have $G_{j}^{*}$ is a bent function and $G_{j}^{**}(y)=G_{j}(-y)=(1-j_{1}-\dots-j_{t})g_{0}(-y)+j_{1}g_{1}(-y)+\dots+j_{t}g_{t}(-y)=(1-j_{1}-\dots-j_{t})g_{0}^{**}(y)+j_{1}g_{1}^{**}(y)+\dots+j_{t}g_{t}^{**}(y)$ and $\mu_{G_{j}^{*}}(y)=u^{-1}, y \in V_{m}$, that is, $g_{s}^{*} (0\leq s \leq t)$ also satisfy the condition of Lemma 2.
Since $F^{*}$ has the same form as $F$ and $g_{s}^{*} (0\leq s \leq t)$ satisfy the condition of Lemma 2, by (9), for any $(a, b)\in V_{r} \times V_{m}$ we have
\begin{equation*}
  \begin{split}
      & W_{F^{*}}(a, b) \\
      &=u^{-1}p^{\frac{m}{2}}\zeta_{p}^{g_{0}^{**}(b)}\zeta_{p^k}^{g(g_{0}^{**}(b)-g_{1}^{**}(b), \dots, g_{0}^{**}(b)-g_{t}^{**}(b))}W_{f_{(g_{0}^{**}(b)-g_{1}^{**}(b), \dots, g_{0}^{**}(b)-g_{t}^{**}(b))}^{*}}(a)\\
      &=u^{-1}p^{\frac{m}{2}}\zeta_{p}^{g_{0}(-b)}\zeta_{p^k}^{g(g_{0}(-b)-g_{1}(-b), \dots, g_{0}(-b)-g_{t}(-b))}W_{f_{(g_{0}(-b)-g_{1}(-b), \dots, g_{0}(-b)-g_{t}(-b))}^{*}}(a).
  \end{split}
\end{equation*}
If for any $y \in V_{m}$, $f_{(g_{0}(y)-g_{1}(y), \dots, g_{0}(y)-g_{t}(y))}^{*}$ is generalized bent, then obviously $F^{*}$ is generalized bent. Suppose $F^{*}$ is generalized bent. If there exists $y \in V_{m}$ such that $f_{(g_{0}(y)-g_{1}(y), \dots, g_{0}(y)-g_{t}(y))}^{*}$ is not generalized bent, let $a \in V_{r}$ with $|W_{f_{(g_{0}(y)-g_{1}(y), \dots, g_{0}(y)-g_{t}(y))}^{*}}(a)|\neq p^{\frac{r}{2}}$ and $b=-y$, then $|W_{F^{*}}(a, b)|\neq p^{\frac{r+m}{2}}$ and $F^{*}$ is not generalized bent, which is a contradiction. Hence, $F^{*}$ is a generalized bent function if and only if for any $y \in V_{m}$, the dual of $f_{(g_{0}(y)-g_{1}(y), \dots, g_{0}(y)-g_{t}(y))}$ is a generalized bent function.
\end{proof}

When $k=1, t=1, m=2$ and $g=0$, $g_{0}(y)=y_{1}y_{2}, g_{1}(y)=y_{1}y_{2}-y_{2}, y=(y_{1}, y_{2}) \in \mathbb{F}_{p}\times \mathbb{F}_{p}$, Theorem 2 reduces to Theorem 3 of \cite{Cesmelioglu5}. The corresponding function of Theorem 3 of \cite{Cesmelioglu5} is in the Generalized Maiorana-McFarland bent functions class (see \cite{Cesmelioglu3}). In \cite{Wang}, the authors showed that the canonical way to construct Generalized Maiorana-McFarland bent functions can be obtained by the generalized indirect sum construction method. In \cite{Wang}, the authors also showed that $p$-ary $PS_{ap}$ bent functions
\begin{equation*}
  g_{s}(y)=Tr_{1}^{m}(\alpha_{s}G(y_{1}y_{2}^{p^{m}-2})), y=(y_{1}, y_{2})\in \mathbb{F}_{p^m} \times \mathbb{F}_{p^m}, 0\leq s\leq t
\end{equation*}
satisfy the condition of Lemma 2 where $m\geq t+1$, $\alpha_{0}, \dots, \alpha_{t}\in \mathbb{F}_{p^m}$ are linearly independent over $\mathbb{F}_{p}$ and $G$ is a permutation over $\mathbb{F}_{p^m}$ with $G(0)=0$.

Since the above $g_{s}(0\leq s \leq t)$ satisfy the condition of Lemma 2 and $\{(g_{0}(y)-g_{1}(y), \dots, g_{0}(y)-g_{t}(y)), y=(y_{1}, y_{2}) \in \mathbb{F}_{p^m} \times \mathbb{F}_{p^m}\}=\mathbb{F}_{p}^{t}$, by Theorem 2 we have the following corollary.
\begin{corollary}
Let $f_{i} (i \in \mathbb{F}_{p}^{t}) : V_{r}\rightarrow \mathbb{Z}_{p^k}$ be generalized bent functions. Let $g_{s} (0\leq s\leq t) : \mathbb{F}_{p^m} \times \mathbb{F}_{p^m} \rightarrow \mathbb{F}_{p}$ be defined as $g_{s}(y)=Tr_{1}^{m}(\alpha_{s}G(y_{1}y_{2}^{p^{m}-2})), y=(y_{1}, y_{2})\in \mathbb{F}_{p^m} \times \mathbb{F}_{p^m}$ where $m\geq t+1$, $\alpha_{0}, \dots, \alpha_{t}\in \mathbb{F}_{p^m}$ are linearly independent over $\mathbb{F}_{p}$ and $G$ is a permutation over $\mathbb{F}_{p^m}$ with $G(0)=0$. Let $g$ be an arbitrary function from $\mathbb{F}_{p}^{t}$ to $\mathbb{Z}_{p^k}$. Then the dual of the generalized bent function $F: V_{r} \times \mathbb{F}_{p^m} \times \mathbb{F}_{p^m}\rightarrow \mathbb{Z}_{p^k}$ defined as $F(x,y)=f_{(g_{0}(y)-g_{1}(y), \dots, g_{0}(y)-g_{t}(y))}(x)+p^{k-1}g_{0}(y)+g(g_{0}(y)-g_{1}(y), \dots, g_{0}(y)-g_{t}(y)), (x, y)=(x, y_{1}, y_{2}) \in V_{r}\times \mathbb{F}_{p^m} \times \mathbb{F}_{p^m}$ is generalized bent if and only if for any $i \in \mathbb{F}_{p}^{t}$, the dual of $f_{i}$ is generalized bent.
\end{corollary}
\section{A property of generalized bent functions whose dual is generalized bent}
In this section, let $f: V_{n}\rightarrow \mathbb{Z}_{p^k}$ be a generalized bent function whose dual $f^{*}$ is generalized bent. By using the knowledge on ideal decomposition in cyclotomic field, we will prove $f^{**}(x)=f(-x)$. For the case of bent functions $f: V_{n}\rightarrow \mathbb{F}_{p}$ whose dual $f^{*}$ is bent, \"{O}zbudak and Pelen \cite{Ozbudak} have shown that $f^{**}(x)=f(-x)$ holds by studying the value distributions of the dual of non-weakly regular bent functions. Compared with the proof method in \cite{Ozbudak}, our proof method seems more concise. Before proof, we need a lemma.
\begin{lemma}
Let $L=\mathbb{Q}(\zeta_{p^k}), \mathcal{O}_{L}=\mathbb{Z}[\zeta_{p^k}]$. Then for any $1\leq j \leq p^{k}-1$, $(1+\zeta_{p^k}^{j})\mathcal{O}_{L}=\mathcal{O}_{L}$ and $(1-\zeta_{p^k}^{j})\mathcal{O}_{L}=((1-\zeta_{p^k})\mathcal{O}_{L})^{p^s}$ where $s=v_{p}(j)$.
\end{lemma}
\begin{proof}
For any $1 \leq j \leq p^k-1$, by
\begin{equation*}
  \frac{1}{1+\zeta_{p^k}^{j}}=\frac{1-\zeta_{p^k}^{j}}{1-\zeta_{p^k}^{2j}}=\frac{1-\zeta_{p^k}^{2lj}}{1-\zeta_{p^k}^{2j}}=\sum_{i=0}^{l-1}\zeta_{p^k}^{2ij} \in \mathcal{O}_{L}
\end{equation*}
where $l\in \mathbb{Z}_{p^k}$ satisfies $2l\equiv 1 (mod \ p^k)$, we have that $1+\zeta_{p^k}^{j}$ is a unit of $\mathcal{O}_{L}$, that is, $(1+\zeta_{p^k}^{j})\mathcal{O}_{L}=\mathcal{O}_{L}$.

For any $1\leq j \leq p^k-1$, let $s=v_{p}(j)$. Then $0\leq s<k$ and $gcd(j, p^{k})=p^{s}$. Let $M=\mathbb{Q}(\zeta_{p^{k-s}})$, then $\mathcal{O}_{M}=\mathbb{Z}[\zeta_{p^{k-s}}]$ and $M$ is a subfield of $L$ since $\zeta_{p^{k-s}}=\zeta_{p^k}^{p^s}$. By

\begin{equation*}
  \frac{1-\zeta_{p^{k-s}}^{\frac{j}{p^{s}}}}{1-\zeta_{p^{k-s}}}=\sum_{i=0}^{\frac{j}{p^s}-1}\zeta_{p^{k-s}}^{i} \in \mathcal{O}_{M}\subseteq \mathcal{O}_{L}
\end{equation*}
and
\begin{equation*}
  \frac{1-\zeta_{p^{k-s}}}{1-\zeta_{p^{k-s}}^{\frac{j}{p^{s}}}}=\frac{1-\zeta_{p^{k-s}}^{\frac{j}{p^s}t}}{1-\zeta_{p^{k-s}}^{\frac{j}{p^{s}}}}=\sum_{i=0}^{t-1}\zeta_{p^{k-s}}^{i\frac{j}{p^s}}\in \mathcal{O}_{M}\subseteq \mathcal{O}_{L}
\end{equation*}
where $t\in \mathbb{Z}_{p^{k-s}}$ satisfies $\frac{j}{p^{s}}t\equiv 1 (mod \ p^{k-s})$, we have that $\frac{1-\zeta_{p^{k-s}}^{\frac{j}{p^{s}}}}{1-\zeta_{p^{k-s}}}$ is a unit of $\mathcal{O}_{L}$. Note that $t$ exists since $gcd(\frac{j}{p^{s}}, p^{k-s})=1$. Hence $(1-\zeta_{p^k}^{j})\mathcal{O}_{L}=(1-\zeta_{p^{k-s}}^{\frac{j}{p^s}})\mathcal{O}_{L}=(1-\zeta_{p^{k-s}})\mathcal{O}_{L}$. By $p\mathcal{O}_{M}=((1-\zeta_{p^{k-s}})\mathcal{O}_{M})^{\varphi(p^{k-s})}$ and $(p\mathcal{O}_{M})\mathcal{O}_{L}=p\mathcal{O}_{L}=((1-\zeta_{p^k})\mathcal{O}_{L})^{\varphi(p^k)}$, we have $((1-\zeta_{p^{k-s}})\mathcal{O}_{L})^{\varphi(p^{k-s})}=((1-\zeta_{p^k})\mathcal{O}_{L})^{\varphi(p^k)}$. By the uniqueness of the decomposition of $(1-\zeta_{p^{k-s}})\mathcal{O}_{L}$, we have $(1-\zeta_{p^{k-s}})\mathcal{O}_{L}=((1-\zeta_{p^k})\mathcal{O}_{L})^{\frac{\varphi(p^k)}{\varphi(p^{k-s})}}=((1-\zeta_{p^k})\mathcal{O}_{L})^{p^{s}}$.
Hence, $(1-\zeta_{p^k}^{j})\mathcal{O}_{L}=((1-\zeta_{p^k})\mathcal{O}_{L})^{p^s}$.
\end{proof}

Now we prove $f^{**}(x)=f(-x)$ by using Lemma 3. In the subsequent of this paper, $L=\mathbb{Q}(\zeta_{p^k}), \mathcal{O}_{L}=\mathbb{Z}[\zeta_{p^k}]$ unless otherwise stated.
\begin{theorem}
Let $f: V_{n}\rightarrow \mathbb{Z}_{p^k}$ be a generalized bent function whose dual $f^{*}$ is also a generalized bent function. Then $f^{**}(x)=f(-x), x \in V_{n}$, where $f^{**}$ is the dual of $f^{*}$.
\end{theorem}
\begin{proof}
Consider the left-hand side of Equation (6),
\begin{equation*}
   \begin{split}
    &\xi \sum_{x \in V_{n}}\varepsilon_{f}(x)\zeta_{p^k}^{f^{*}(x)}\zeta_{p}^{\langle a, x\rangle}\\
    & = \xi \sum_{x\in V_{n}:\ \varepsilon_{f}(x)=1}\zeta_{p^k}^{f^{*}(x)}\zeta_{p}^{\langle a, x\rangle}-\xi \sum_{x \in V_{n}: \ \varepsilon_{f}(x)=-1}\zeta_{p^k}^{f^{*}(x)}\zeta_{p}^{\langle a, x\rangle} \\
    & = 2\xi \sum_{x\in V_{n}: \ \varepsilon_{f}(x)=1}\zeta_{p^k}^{f^{*}(x)}\zeta_{p}^{\langle a, x\rangle}-\xi \sum_{x \in V_{n}}\zeta_{p^k}^{f^{*}(x)}\zeta_{p}^{\langle a, x\rangle}\\
    & =2\xi \sum_{x\in V_{n}: \ \varepsilon_{f}(x)=1}\zeta_{p^k}^{f^{*}(x)}\zeta_{p}^{\langle a, x\rangle}-\xi W_{f^{*}}(-a),
    \end{split}
\end{equation*}
where $\xi=1$ if $p \equiv 1 \ (mod \ 4)$ or $n$ is even and $\xi =\sqrt{-1}$ if $p \equiv 3 \ (mod \ 4)$ and $n$ is odd, $\varepsilon_{f}: V_{n}\rightarrow \{\pm 1\}$ is defined by (5). Since $f^{*}$ is a generalized bent function, we have $W_{f^{*}}(-a)=\xi \varepsilon_{f^{*}}(-a)p^{\frac{n}{2}}\zeta_{p^k}^{f^{**}(-a)}$ where $\varepsilon_{f^{*}}: V_{n}\rightarrow \{\pm 1\}$. Thus, for any $a \in V_{n}$ we have
\begin{equation}\label{10}
\xi \sum_{x \in V_{n}}\varepsilon_{f}(x)\zeta_{p^k}^{f^{*}(x)}\zeta_{p}^{\langle a, x\rangle}
=2\xi \sum_{x\in V_{n}: \ \varepsilon_{f}(x)=1}\zeta_{p^k}^{f^{*}(x)}\zeta_{p}^{\langle a, x\rangle}-\xi^{2}\varepsilon_{f^{*}}(-a)p^{\frac{n}{2}}\zeta_{p^{k}}^{f^{**}(-a)}.
\end{equation}
Let $X_{a}=\xi \sum_{x\in V_{n}: \ \varepsilon_{f}(x)=1}\zeta_{p^k}^{f^{*}(x)}\zeta_{p}^{\langle a, x\rangle}$. By (6) and (10), we have
\begin{equation}\label{11}
  2X_{a}=p^{\frac{n}{2}}\zeta_{p^k}^{f(a)}(1+\xi^{2}\varepsilon_{f^{*}}(-a)\zeta_{p^k}^{f^{**}(-a)-f(a)}), a \in V_{n}.
\end{equation}
Suppose there exists $a \in V_{n}$ such that $f^{**}(-a)\neq f(a)$, that is, $\Delta_{a}\triangleq f^{**}(-a)-f(a)\neq 0$.

1) When $n$ is even, then $\xi=1$. Note that in this case, $X_{a} \in \mathcal{O}_{L}$. By (11) and Lemma 3, we have
\begin{equation*}
  (2\mathcal{O}_{L})(X_{a}\mathcal{O}_{L})=\left\{\begin{split}
                                                      & ((1-\zeta_{p^k})\mathcal{O}_{L})^{\varphi(p^k)\frac{n}{2}} \ \  \ \  if \ \varepsilon_{f^{*}}(-a)=1,\\
                                                      & ((1-\zeta_{p^k})\mathcal{O}_{L})^{\varphi(p^k)\frac{n}{2}+p^{v_{p}(\Delta_{a})}} \ if \ \varepsilon_{f^{*}}(-a)=-1.
                                                  \end{split}\right.
\end{equation*}
Since $\frac{1}{2} \notin \mathcal{O}_{L}$, we have $2\mathcal{O}_{L}\neq \mathcal{O}_{L}$. Indeed, if $\frac{1}{2} \in \mathcal{O}_{L}$, then by $\{\zeta_{p^k}^{0}, \dots, $ $\zeta_{p^k}^{\varphi(p^k)-1}\}$ is an integer basis of $\mathcal{O}_{L}$, $\frac{1}{2}=\sum_{i=0}^{\varphi(p^k)-1}a_{i}\zeta_{p^k}^{i}$ where $a_{i}\in \mathbb{Z}$, that is, $\sum_{i=0}^{\varphi(p^k)-1}2a_{i}\zeta_{p^k}^{i}=\zeta_{p^k}^{0}$. This equation deduces $a_{0}=\frac{1}{2}, a_{i}=0 (1\leq i \leq \varphi(p^k)-1)$, which contradicts $a_{0}\in \mathbb{Z}$. Then by the uniqueness of the decomposition of $2\mathcal{O}_{L}$, we have
\begin{equation}\label{12}
  2\mathcal{O}_{L}=((1-\zeta_{p^k})\mathcal{O}_{L})^{t}
\end{equation}
for some positive integer $t$. From (12), we have $2\mathbb{Z}=\mathbb{Z}\cap (1-\zeta_{p^k})\mathcal{O}_{L}$. And from $p\mathcal{O}_{L}=((1-\zeta_{p^k})\mathcal{O}_{L})^{\varphi(p^k)}$, we have $p\mathbb{Z}=\mathbb{Z}\cap (1-\zeta_{p^k})\mathcal{O}_{L}$. Hence, $2\mathbb{Z}=p\mathbb{Z}$, which is a contradiction since $p$ is an odd prime. So in this case, $f^{**}(-a)=f(a)$ for any $a \in V_{n}$.

2) When $n$ is odd, by multiplying both sides of (11) by $\sqrt{p}$, we obtain
\begin{equation}\label{13}
  2X_{a}\sqrt{p}=p^{\frac{n+1}{2}}\zeta_{p^k}^{f(a)}(1+\xi^{2}\varepsilon_{f^{*}}(-a)\zeta_{p^k}^{f^{**}(-a)-f(a)}).
\end{equation}
Recall that when $n$ is odd, $\xi=1$ if $p\equiv 1 \ (mod \ 4)$ and $\xi=\sqrt{-1}$ if $p\equiv 3 \ (mod \ 4)$. Note that $\xi^{2}\in \{\pm1\}$ and $X_{a}\sqrt{p}\in \mathcal{O}_{L}$ since $\xi \sqrt{p}=\sum_{i \in \mathbb{F}_{p}^{*}}\eta(i)\zeta_{p}^{i} \in \mathbb{Z}[\zeta_{p}]$ by a well known result on Gauss sums (see \cite{Lidl}). By (13) and Lemma 3, we have
\begin{equation*}
  (2\mathcal{O}_{L})((X_{a}\sqrt{p})\mathcal{O}_{L})=\left\{\begin{split}
                                                      & ((1-\zeta_{p^k})\mathcal{O}_{L})^{\varphi(p^k)\frac{n+1}{2}} \ \  \ \  if \ \xi^{2}\varepsilon_{f^{*}}(-a)=1,\\
                                                      & ((1-\zeta_{p^k})\mathcal{O}_{L})^{\varphi(p^k)\frac{n+1}{2}+p^{v_{p}(\Delta_{a})}} \ if \ \xi^{2}\varepsilon_{f^{*}}(-a)=-1.
                                                  \end{split}\right.
\end{equation*}
Since $2\mathcal{O}_{L}\neq \mathcal{O}_{L}$, by the uniqueness of the decomposition of $2\mathcal{O}_{L}$, we have $2\mathcal{O}_{L}=((1-\zeta_{p^k})\mathcal{O}_{L})^{t}$ for some positive integer $t$. Then with the same argument as 1), we have $2\mathbb{Z}=p\mathbb{Z}$, which is a contradiction. So in this case, $f^{**}(-a)=f(a)$ for any $a \in V_{n}$.
\end{proof}
\section{Nonexistence of self-dual generalized bent function if $p\equiv 3 \ (mod \ 4)$ and $n$ is odd}
In this section, we will show that there is no self-dual generalized bent function $f: V_{n}\rightarrow \mathbb{Z}_{p^k}$ if $p\equiv 3 \ (mod \ 4)$ and $n$ is odd. Note that in \cite{Cesmelioglu4}, for $p \equiv 3 \ (mod \ 4)$ and $n$ is odd, the authors showed that there is no weakly regular self-dual bent function and they pointed out that there exist non-weakly regular self-dual bent functions. Indeed, they pointed out that the non-weakly regular bent function $g_{3}(x)=Tr_{1}^{3}(x^{22}+x^{8})$ from $\mathbb{F}_{3^3}$ to $\mathbb{F}_{3}$ is self-dual. But in fact, it is easy to verify that $g_{3}$ is not self-dual by using MAGMA and there is no self-dual bent function from $V_{n}$ to $\mathbb{F}_{p}$ if $p\equiv 3 \ (mod \ 4)$ and $n$ is odd according to the theorem of this section.

Let $f: V_{n}\rightarrow \mathbb{Z}_{p^k}$ be a generalized bent function satisfying $f(x)=f(-x), x \in V_{n}$. For any $a \in V_{n}$, we have
\begin{equation*}
  \begin{split}
    W_{f}(a) & =\sum_{x \in V_{n}}\zeta_{p^k}^{f(x)}\zeta_{p}^{-\langle a, x\rangle} \\
             & =\sum_{x \in V_{n}}\zeta_{p^k}^{f(-x)}\zeta_{p}^{\langle a, x\rangle} \\
             & =\sum_{x \in V_{n}}\zeta_{p^k}^{f(x)}\zeta_{p}^{\langle a, x\rangle}\\
             &=W_{f}(-a),
  \end{split}
\end{equation*}
where in the second equation we use the change of variable $x\mapsto -x$ and in the third equation we use $f(x)=f(-x)$. By $W_{f}(a)=W_{f}(-a)$ and $W_{f}(a)=\xi \varepsilon_{f}(a)p^{\frac{n}{2}}\zeta_{p^k}^{f^{*}(a)}$ where $\xi=1$ if $p\equiv 1\ (mod \ 4)$ or $n$ is even and $\xi=\sqrt{-1}$ if $p\equiv 3 \ (mod \ 4)$ and $n$ is odd, $\varepsilon_{f}: V_{n}\rightarrow \{\pm 1\}$, we have
\begin{equation}\label{14}
\varepsilon_{f}(a)=\varepsilon_{f}(-a), f^{*}(a)=f^{*}(-a).
\end{equation}
Note that for any generalized bent function $f: V_{n}\rightarrow \mathbb{Z}_{p^k}$ satisfying $f(x)=f(-x)$, Equation (14) holds.

For any $a\in V_{n}$, let
\begin{equation*}
  S_{a}=\sum_{x\in V_{n}: \ \varepsilon_{f}(x)=1}\zeta_{p^k}^{f^{*}(x)}\zeta_{p}^{\langle a, x\rangle}, \ T_{a}=\sum_{x\in V_{n}: \ \varepsilon_{f}(x)=-1}\zeta_{p^k}^{f^{*}(x)}\zeta_{p}^{\langle a, x\rangle}.
\end{equation*}
Suppose the dual $f^{*}$ is generalized bent. Then by the definitions of $S_{a}$ and $T_{a}$, we have $S_{a}+T_{a}=W_{f^{*}}(-a)=\xi \varepsilon_{f^{*}}(-a)p^{\frac{n}{2}}\zeta_{p^k}^{f^{**}(-a)}$ where $\varepsilon_{f^{*}}: V_{n}\rightarrow \{\pm 1\}$. Since $f^{*}(x)=f^{*}(-x)$, we have $\varepsilon_{f^{*}}(-a)=\varepsilon_{f^{*}}(a)$. By Theorem 3, we have $f^{**}(-a)=f(a)$. Hence, we obtain

\begin{equation}\label{15}
  S_{a}+T_{a}=\xi \varepsilon_{f^{*}}(a)p^{\frac{n}{2}}\zeta_{p^k}^{f(a)}.
\end{equation}
By (6), we have
\begin{equation}\label{16}
  S_{a}-T_{a}=\xi^{-1}p^{\frac{n}{2}}\zeta_{p^k}^{f(a)}.
\end{equation}
By (15) and (16), we have
\begin{equation}\label{17}
  S_{a}=p^{\frac{n}{2}}\zeta_{p^k}^{f(a)}\frac{\xi \varepsilon_{f^{*}}(a)+\xi^{-1}}{2}, \ T_{a}=p^{\frac{n}{2}}\zeta_{p^k}^{f(a)}\frac{\xi \varepsilon_{f^{*}}(a)-\xi^{-1}}{2}.
\end{equation}

Based on the above analysis and (17), we obtain the following proposition. For weakly regular generalized bent functions, the following proposition is obvious. But for non-weakly regular generalized bent functions, the following proposition is not obvious.

\begin{proposition}
Let $f: V_{n}\rightarrow \mathbb{Z}_{p^k}$ be a generalized bent function which satisfies that $f(x)=f(-x), x \in V_{n}$ and the dual $f^{*}$ is a generalized bent function. Let $\varepsilon_{f}, \varepsilon_{f^{*}}: V_{n}\rightarrow \{\pm 1\}$ be defined by (5). Then

1) If $p \equiv 1 \ (mod \ 4)$ or $n$ is even, then $\varepsilon_{f^{*}}(0)=\varepsilon_{f}(0)$;

2) If $p \equiv 3 \ (mod \ 4)$ and $n$ is odd, then $\varepsilon_{f^{*}}(0)=-\varepsilon_{f}(0)$.
\end{proposition}
\begin{proof}
Let $g_{0}: V_{n}\rightarrow \mathbb{F}_{p}$ and $\bar{g_{1}}: V_{n}\rightarrow \mathbb{Z}_{p^{k-1}}$ satisfy $f^{*}=g_{0}p^{k-1}+\bar{g_{1}}$. Note that when $k=1$, $\mathbb{Z}_{p^{k-1}}=\{0\}$ and $g_{0}=f^{*}$, $\bar{g_{1}}=0$.

1) If $p \equiv 1 \ (mod \ 4)$ or $n$ is even, then $\xi=1$. By (17), we have $T_{a}=0$ if $\varepsilon_{f^{*}}(a)=1$ and $S_{a}=0$ if $\varepsilon_{f^{*}}(a)=-1$. Suppose $\varepsilon_{f^{*}}(0)=-\varepsilon_{f}(0)$. Without loss of generality, assume $\varepsilon_{f}(0)=1$. Then $\varepsilon_{f^{*}}(0)=-1$ and $S_{0}=0$. By the definition of $S_{0}$, we have $\sum_{x\in V_{n}: \ \varepsilon_{f}(x)=1}\zeta_{p^k}^{f^{*}(x)}=0$. For any $i_{0} \in \mathbb{F}_{p}, \bar{i_{1}} \in \mathbb{Z}_{p^{k-1}}$, let
\begin{equation}\label{18}
  A_{i_{0}, \bar{i_{1}}}=\{x \in V_{n}: \ \varepsilon_{f}(x)=1 \ and \ g_{0}(x)=i_{0}, \bar{g_{1}}(x)= \bar{i_{1}}\}
\end{equation}
and $N_{i_{0}, \bar{i_{1}}}$ denote the size of $A_{i_{0}, \bar{i_{1}}}$. Then we have
\begin{equation*}
  \sum_{i_{0}\in \mathbb{F}_{p}}\sum_{\bar{i_{1}} \in \mathbb{Z}_{p^{k-1}}}N_{i_{0}, \bar{i_{1}}}\zeta_{p}^{i_{0}}\zeta_{p^k}^{\bar{i_{1}}}=0.
\end{equation*}
By $\sum_{i=0}^{p-1}\zeta_{p}^{i}=0$ and the above equation we obtain
\begin{equation*}
  \sum_{i_{0}=0}^{p-2}\sum_{\bar{i_{1}} \in \mathbb{Z}_{p^{k-1}}}(N_{i_{0}, \bar{i_{1}}}-N_{p-1, \bar{i_{1}}})\zeta_{p}^{i_{0}}\zeta_{p^k}^{\bar{i_{1}}}=0,
\end{equation*}
which deduces $N_{i_{0}, \bar{i_{1}}}=N_{p-1, \bar{i_{1}}}$ for any $i_{0}\in \mathbb{F}_{p}, \bar{i_{1}} \in \mathbb{Z}_{p^{k-1}}$ since $\{\zeta_{p}^{i}\zeta_{p^k}^{j}: 0\leq i \leq p-2, 0\leq j \leq p^{k-1}-1\}$ is an integral basis of $\mathbb{Z}[\zeta_{p^k}]$. In particular, $N_{g_{0}(0), \bar{g_{1}}(0)}=N_{i_{0}, \bar{g_{1}}(0)}$ for any $i_{0} \in \mathbb{F}_{p}$.

By (14), we have $f^{*}(x)=f^{*}(-x)$ and $\varepsilon_{f}(x)=\varepsilon_{f}(-x)$. From $f^{*}(x)=f^{*}(-x)$, we have $g_{0}(x)=g_{0}(-x)$ and $\bar{g_{1}}(x)=\bar{g_{1}}(-x)$. And by $\varepsilon_{f}(x)=\varepsilon_{f}(-x)$, we have $x\in A_{i_{0}, \bar{i_{1}}}$ if and only if $-x\in A_{i_{0}, \bar{i_{1}}}$. Note that $0 \in A_{g_{0}(0), \bar{g_{1}}(0)}$ since $\varepsilon_{f}(0)=1$. Hence, $N_{g_{0}(0),\bar{g_{1}}(0)}$ is odd and $N_{i_{0}, \bar{i_{1}}}$ is even if $(i_{0}, \bar{i_{1}})\neq (g_{0}(0),\bar{g_{1}}(0))$, which contradicts $N_{g_{0}(0), \bar{g_{1}}(0)}=N_{i_{0}, \bar{g_{1}}(0)}, i_{0} \in \mathbb{F}_{p}$. Hence $\varepsilon_{f^{*}}(0)=\varepsilon_{f}(0)$.

2) If $p \equiv 3 \ (mod \ 4)$ and $n$ is odd, then $\xi=\sqrt{-1}$. By (17), we have $S_{a}=0$ if $\varepsilon_{f^{*}}(a)=1$ and $T_{a}=0$ if $\varepsilon_{f^{*}}(a)=-1$. Suppose $\varepsilon_{f^{*}}(0)=\varepsilon_{f}(0)$. Without loss of generality, assume $\varepsilon_{f}(0)=1$. Then $\varepsilon_{f^{*}}(0)=1$ and $S_{0}=0$. By the definition of $S_{0}$, we have $\sum_{x\in V_{n}: \ \varepsilon_{f}(x)=1}\zeta_{p^k}^{f^{*}(x)}=0$. For any $i_{0} \in \mathbb{F}_{p}, \bar{i_{1}} \in \mathbb{Z}_{p^{k-1}}$, let $A_{i_{0}, \bar{i_{1}}}$ be defined by (18) and $N_{i_{0}, \bar{i_{1}}}$ denote the size of $A_{i_{0}, \bar{i_{1}}}$. Then with the same argument as 1), we have that $N_{g_{0}(0), \bar{g_{1}}(0)}=N_{i_{0}, \bar{g_{1}}(0)}$ for any $i_{0}\in \mathbb{F}_{p}$ and $N_{g_{0}(0), \bar{g_{1}}(0)}$ is odd, $N_{i_{0}, \bar{i_{1}}}$ is even if $(i_{0}, \bar{i_{1}})\neq (g_{0}(0), \bar{g_{1}}(0))$, which is a contradiction. Hence, $\varepsilon_{f^{*}}(0)=-\varepsilon_{f}(0)$.
\end{proof}

By using Proposition 1, we obtain the following theorem.
\begin{theorem}
There is no self-dual generalized bent function $f: V_{n}\rightarrow \mathbb{Z}_{p^k}$ if $p\equiv 3 \ (mod \ 4)$ and $n$ is odd.
\end{theorem}
\begin{proof}
Let $f: V_{n}\rightarrow \mathbb{Z}_{p^k}$ be a self-dual generalized bent function. Since $f$ is a generalized bent function whose dual is generalized bent, we have $f^{**}(x)=f(-x)$ by Theorem 3. And since $f$ is self-dual, that is, $f^{*}=f$, we have $f(x)=f(-x)$. By $f^{*}=f$ and Proposition 1, we have $\varepsilon_{f}(0)=-\varepsilon_{f}(0)$ if $p \equiv 3 \ (mod \ 4)$ and $n$ is odd. Hence $\varepsilon_{f}(0)=0$, which contradicts $\varepsilon_{f}(0) \in \{\pm 1\}$. Therefore, there is no self-dual generalized bent function $f: V_{n}\rightarrow \mathbb{Z}_{p^k}$ if $p\equiv 3 \ (mod \ 4)$ and $n$ is odd.
\end{proof}
\section{A secondary construction of self-dual generalized bent functions if $p \equiv 1 \ (mod \ 4)$ or $n$ is even}
In this section, for $p\equiv 1 \ (mod \ 4)$ or $n$ is even, we give a secondary construction of self-dual generalized bent functions $f: V_{n}\rightarrow \mathbb{Z}_{p^k}$. First, we give a lemma.
\begin{lemma}
Let $m$ be a positive integer and $m$ be even if $p\equiv 3 \ (mod \ 4)$. Let $a \in \mathbb{F}_{p^m}^{*}$ be an arbitrary element. Let $\alpha, \beta \in \{\pm z^{\frac{p^m-1}{4}}\}$ where $z$ is a primitive element of $\mathbb{F}_{p^m}$. Let $f_{i} (i \in \mathbb{F}_{p}): V_{r}\rightarrow \mathbb{Z}_{p^k}$ be generalized bent functions. Let $g_{0}: \mathbb{F}_{p^m}\times \mathbb{F}_{p^m}\rightarrow \mathbb{F}_{p}$ be defined as $g_{0}(y_{1}, y_{2})=Tr_{1}^{m}(\frac{\beta}{2}(y_{1}^{2}+y_{2}^{2})), (y_{1}, y_{2})\in \mathbb{F}_{p^m}\times \mathbb{F}_{p^m}$. Let $h: \mathbb{F}_{p^m}\rightarrow \mathbb{F}_{p}$ and $g: \mathbb{F}_{p}\rightarrow \mathbb{Z}_{p^k}$ be arbitrary functions. Then the function $F(x, y_{1}, y_{2})=f_{h(a\alpha y_{1}+a y_{2})}(x)+p^{k-1}g_{0}(y_{1}, y_{2})+g(h(a\alpha y_{1}+a y_{2}))$ is a generalized bent function from $V_{r} \times \mathbb{F}_{p^m}\times \mathbb{F}_{p^m}$ to $\mathbb{Z}_{p^k}$ and its dual $F^{*}(x, y_{1}, y_{2})=f_{h(-\beta(a\alpha y_{1}+a y_{2}))}^{*}(x)+p^{k-1}g_{0}(y_{1}, y_{2})+g(h(-\beta(a\alpha y_{1}+a y_{2})))$.
\end{lemma}
\begin{proof}
First note that $4  \mid (p^{m}-1)$. Indeed, if $p\equiv 1 \ (mod \ 4)$, then $p^m\equiv 1 \ (mod \ 4)$ and if $p\equiv 3 \ (mod \ 4)$ and $m$ is even, then $p^{m}\equiv (-1)^{m}\equiv 1 \ (mod \ 4)$, that is, $4 \mid (p^{m}-1)$.

By Lemma 2 and (9), we only need to prove that $g_{0}$ and $g_{1}$ defined by $g_{1}(y_{1}, y_{2})=g_{0}(y_{1}, y_{2})-h(a\alpha y_{1}+a y_{2}), (y_{1}, y_{2}) \in \mathbb{F}_{p^m} \times \mathbb{F}_{p^m}$ satisfy the condition of Lemma 2 and $g_{0}^{*}(y_{1}, y_{2})=g_{0}(y_{1}, y_{2})$, $g_{1}^{*}(y_{1}, y_{2})=g_{0}(y_{1}, y_{2})-h(-\beta(a\alpha y_{1}+ay_{2})), (y_{1}, y_{2})\in \mathbb{F}_{p^m}\times \mathbb{F}_{p^m}$. For any $(b_{1}, b_{2})\in \mathbb{F}_{p^m} \times \mathbb{F}_{p^m}$, we have
\begin{equation}\label{19}
  \begin{split}
      & W_{g_{1}}(b_{1}, b_{2}) \\
      & =\sum_{y_{1}, y_{2} \in \mathbb{F}_{p^m}}\zeta_{p}^{Tr_{1}^{m}(\frac{\beta}{2}(y_{1}^{2}+y_{2}^{2}))-h(a\alpha y_{1}+a y_{2})-Tr_{1}^{m}(b_{1}y_{1}+b_{2}y_{2})}\\
      & =\sum_{z_{1}, z_{2} \in \mathbb{F}_{p^m}}\zeta_{p}^{Tr_{1}^{m}(-\frac{\alpha\beta}{2a^{2}}z_{1}z_{2})-h(z_{2})-Tr_{1}^{m}(\frac{b_{1}-\alpha b_{2}}{2a}z_{1}+\frac{-\alpha b_{1}+b_{2}}{2a}z_{2})}\\
      & =W_{Tr_{1}^{m}(-\frac{\alpha\beta}{2a^{2}}z_{1}z_{2})-h(z_{2})}(\frac{b_{1}-\alpha b_{2}}{2a}, \frac{-\alpha b_{1}+b_{2}}{2a})
  \end{split}
\end{equation}
where in the second equation we use the change of variables $z_{1}=ay_{1}+a\alpha y_{2}$, $z_{2}=a\alpha y_{1}+ay_{2}$ and $\alpha^{2}=-1$.
Since $Tr_{1}^{m}(-\frac{\alpha\beta}{2a^{2}}z_{1}z_{2})-h(z_{2})$ is an Maiorana-McFarland bent function, we have
\begin{equation}\label{20}
  \begin{split}
     & W_{Tr_{1}^{m}(-\frac{\alpha\beta}{2a^{2}}z_{1}z_{2})-h(z_{2})}(\frac{b_{1}-\alpha b_{2}}{2a}, \frac{-\alpha b_{1}+b_{2}}{2a})\\
     & =p^{m}\zeta_{p}^{Tr_{1}^{m}(\frac{2a^{2}}{\alpha\beta}\cdot\frac{b_{1}-\alpha b_{2}}{2a}\cdot\frac{-\alpha b_{1}+b_{2}}{2a})-h(-\frac{2a^{2}}{\alpha\beta}\cdot\frac{b_{1}-\alpha b_{2}}{2a})}\\
     & =p^{m}\zeta_{p}^{Tr_{1}^{m}(\frac{\beta}{2}(b_{1}^{2}+b_{2}^{2}))-h(-\beta(a\alpha b_{1}+ab_{2}))}
   \end{split}
\end{equation}
where in the last equation we use $\alpha^{2}=\beta^{2}=-1$. Note that $h$ is arbitrary. Hence, by (19) and (20) we have that $g_{0}, g_{1}$ are regular bent functions and $g_{0}^{*}(y_{1}, y_{2})=g_{0}(y_{1}, y_{2})$, $g_{1}^{*}(y_{1}, y_{2})=g_{0}(y_{1}, y_{2})-h(-\beta(a\alpha y_{1}+ay_{2})), (y_{1}, y_{2})\in \mathbb{F}_{p^m}\times \mathbb{F}_{p^m}$.

By (19) and (20), $(1-i)g_{0}(y_{1}, y_{2})+ig_{1}(y_{1}, y_{2})=g_{0}(y_{1}, y_{2})-ih(a\alpha y_{1}+ay_{2})$ is a regular bent function and $((1-i)g_{0}+ig_{1})^{*}(y_{1}, y_{2})=g_{0}(y_{1}, y_{2})-ih(-\beta(a\alpha y_{1}+ay_{2}))$, $(1-i)g_{0}^{*}(y_{1}, y_{2})+ig_{1}^{*}(y_{1}, y_{2})=(1-i)g_{0}(y_{1}, y_{2})+ig_{0}(y_{1}, y_{2})-ih(-\beta(a\alpha y_{1}+ay_{2}))=g_{0}(y_{1}, y_{2})-ih(-\beta(a\alpha y_{1}+ay_{2}))$ for any $i \in \mathbb{F}_{p}$, that is, $g_{0}, g_{1}$ satisfy the condition of Lemma 2.
\end{proof}

The following Theorem gives a secondary construction of self-dual generalized bent functions.
\begin{theorem}
With the same notation as Lemma 4. The function $F$ constructed by Lemma 4 is a self-dual generalized bent function if any one of the following conditions is satisfied:

1) $p\equiv 1 \ (mod \ 4)$, $f_{i} (i \in \mathbb{F}_{p})$ are self-dual generalized bent functions satisfying $f_{i}=f_{j}$ if $i=j \beta^{a}$ for some $0\leq a \leq 3$, $h: \mathbb{F}_{p^m} \rightarrow \mathbb{F}_{p}$ is defined as $h(x)=Tr_{1}^{m}(x)$, $g: \mathbb{F}_{p}\rightarrow \mathbb{Z}_{p^k}$ is an arbitrary function satisfying $g(y)=g(y')$ if $y=y'\beta^{b}$ for some $0\leq b \leq 3$.

2) $p\equiv 1 \ (mod \ 4)$ or $m$ is even, $f_{i} (i \in \mathbb{F}_{p})$ are self-dual generalized bent functions satisfying $f_{i}=f_{-i}$, $h: \mathbb{F}_{p^m}\rightarrow \mathbb{F}_{p}$ is defined as $h(x)=Tr_{1}^{m}(x^{2})$, $g: \mathbb{F}_{p}\rightarrow \mathbb{Z}_{p^k}$ is an arbitrary function satisfying $g(y)=g(-y)$.

3) $p\equiv 1 \ (mod \ 4)$ or $m$ is even, $f_{i} (i \in \mathbb{F}_{p})$ are self-dual generalized bent functions, $h: \mathbb{F}_{p^m}\rightarrow \mathbb{F}_{p}$ is defined as $h(x)=Tr_{1}^{m}(x^{4})$, $g: \mathbb{F}_{p}\rightarrow \mathbb{Z}_{p^k}$ is an arbitrary function.
\end{theorem}
\begin{proof}
1) If $p\equiv 1 \ (mod \ 4)$, that is, $4 \mid (p-1)$, we have $\beta^{p-1}=(\pm z^{\frac{p^m-1}{4}})^{p-1}=(z^{\frac{p-1}{4}})^{p^m-1}=1$, that is, $\beta \in \mathbb{F}_{p}$. When $h(x)=Tr_{1}^{m}(x)$, for any $(y_{1}, y_{2})\in \mathbb{F}_{p^m}\times \mathbb{F}_{p^m}$, $h(-\beta(a\alpha y_{1}+ay_{2}))=Tr_{1}^{m}(-\beta(a\alpha y_{1}+ay_{2}))=-\beta Tr_{1}^{m}(a\alpha y_{1}+ay_{2})=-\beta h(a\alpha y_{1}+ay_{2})$. Since $f_{i} (i \in \mathbb{F}_{p})$ are self-dual generalized bent functions and $f_{i}=f_{j}$ if $i=j \beta^{a}$ for some $0\leq a \leq 3$, $\beta^{2}=-1$, we have $f_{h(-\beta(a\alpha y_{1}+ay_{2}))}^{*}=f_{h(-\beta(a\alpha y_{1}+ay_{2}))}=f_{-\beta h(a\alpha y_{1}+ay_{2})}=f_{h(a\alpha y_{1}+ay_{2})}$. Since $g(y)=g(y')$ if $y'=y \beta^{b}$ for some $0 \leq b \leq 3$ and $\beta^{2}=-1$, we have $g(h(-\beta(a\alpha y_{1}+ay_{2})))=g(-\beta h(a\alpha y_{1}+ay_{2}))=g(h(a\alpha y_{1}+ay_{2}))$. Hence, it is easy to see that the generalized bent function $F$ constructed by Lemma 4 satisfies $F=F^{*}$, that is, $F$ is a self-dual generalized bent function.

2) When $h(x)=Tr_{1}^{m}(x^{2})$, for any $(y_{1}, y_{2})\in \mathbb{F}_{p^m}\times \mathbb{F}_{p^m}$, since $\beta^{2}=-1$, we have $h(-\beta(a\alpha y_{1}+ay_{2}))=Tr_{1}^{m}((-\beta(a\alpha y_{1}+ay_{2}))^{2})=Tr_{1}^{m}(-(a\alpha y_{1}+ay_{2})^{2})=-Tr_{1}^{m}((a\alpha y_{1}+ay_{2})^{2})=-h(a\alpha y_{1}+ay_{2})$. Since $f_{i} (i \in \mathbb{F}_{p})$ are self-dual generalized bent functions and $f_{i}=f_{-i}$, we have $f_{h(-\beta(a\alpha y_{1}+ay_{2}))}^{*}=f_{h(-\beta(a\alpha y_{1}+ay_{2}))}=f_{-h(a\alpha y_{1}+ay_{2})}=f_{h(a\alpha y_{1}+ay_{2})}$. Since $g(y)=g(-y)$, we have $g(h(-\beta(a\alpha y_{1}+ay_{2})))=g(-h(a\alpha y_{1}+ay_{2}))=g(h(a\alpha y_{1}+ay_{2}))$. Hence, it is easy to see that the generalized bent function $F$ constructed by Lemma 4 satisfies $F=F^{*}$, that is, $F$ is a self-dual generalized bent function.

3) When $h(x)=Tr_{1}^{m}(x^{4})$, for any $(y_{1}, y_{2})\in \mathbb{F}_{p^m}\times \mathbb{F}_{p^m}$, since $\beta^{4}=1$, we have $h(-\beta(a\alpha y_{1}+ay_{2}))=Tr_{1}^{m}((-\beta(a\alpha y_{1}+ay_{2}))^{4})=Tr_{1}^{m}((a\alpha y_{1}+ay_{2})^{4})=h(a\alpha y_{1}+ay_{2})$. Since $f_{i} (i \in \mathbb{F}_{p})$ are self-dual generalized bent functions, we have $f_{h(-\beta(a\alpha y_{1}+ay_{2}))}^{*}=f_{h(-\beta(a\alpha y_{1}+ay_{2}))}=f_{h(a\alpha y_{1}+ay_{2})}$. For an arbitrary function $g: \mathbb{F}_{p}\rightarrow \mathbb{Z}_{p^k}$, $g(h(-\beta(a\alpha y_{1}+ay_{2})))=g(h(a\alpha y_{1}+ay_{2}))$. Hence, it is easy to see that the generalized bent function $F$ constructed by Lemma 4 satisfies $F=F^{*}$, that is, $F$ is a self-dual generalized bent function.
\end{proof}
\begin{remark}
One can verify that Theorem 3 of \cite{Cesmelioglu4} is a special case of the above case 1) with $k=1, m=1$ and $g=0$. In \cite{Cesmelioglu4}, for $p\equiv 3 \ (mod \ 4)$ and $n$ is even, the authors only illustrated that there exist ternary non-quadratic self-dual bent functions from $V_{n}$ to $\mathbb{F}_{p}$ by considering special ternary bent functions. Theorem 5 can be used to construct non-quadratic self-dual bent functions from $V_{n}$ to $\mathbb{F}_{p}$ for $p\equiv 3 \ (mod \ 4)$ and even integer $n\geq 6$ by using (non)-quadratic self-dual bent functions as building blocks.
\end{remark}

We give two examples by using Theorem 5.

\begin{example}
Let $p=5, k=2, r=1, m=1$. Let $\alpha=\beta=2$ and $a=1$. Let $f_{i}(x)=5x^{2}, x \in \mathbb{F}_{5}, i=0, 2, 3, 4$ and $f_{1}(x)=20x^{2}, x \in \mathbb{F}_{5}$. Then $f_{i} (i \in \mathbb{F}_{5})$ are self-dual generalized bent functions. Let $h(x)=Tr_{1}^{m}(x^{4})=x^{4}, x \in \mathbb{F}_{5}$. Let $g: \mathbb{F}_{5}\rightarrow \mathbb{Z}_{5^{2}}$ be defined as $g(y)=2y^{2}, y \in \mathbb{F}_{5}$. Then the generalized bent function constructed by Lemma 4 is $F(x, y_{1}, y_{2})=f_{(2y_{1}+y_{2})^{4}}(x)+5(y_{1}^{2}+y_{2}^{2})+2((2y_{1}+y_{2})^{4} mod \ 5)^{2}, (x, y_{1}, y_{2}) \in \mathbb{F}_{5}^{3}$ and it is a self-dual generalized bent function according to 3) of Theorem 5.
\end{example}

\begin{example}
Let $p=7, k=1, r=m=2$. Let $z$ be the primitive element of $\mathbb{F}_{7^2}$ with $z^{2}+6z+3=0$. Let $\alpha=\beta=z^{12}$ and $a=z$. Let $f_{0}(x)=Tr_{1}^{2}(4z^{12}x^{2})$, $f_{1}(x)=f_{6}(x)=Tr_{1}^{2}(3z^{12}x^{2})+1$, $f_{2}(x)=f_{5}(x)=Tr_{1}^{2}(3z^{12}x^{2})+2$, $f_{3}(x)=f_{4}(x)=Tr_{1}^{2}(3z^{12}x^{2})+3$, $x \in \mathbb{F}_{7^2}$. Then $f_{i} (i \in \mathbb{F}_{7})$ are quadratic self-dual bent functions. Let $h(x)=Tr_{1}^{2}(x^{2}), x \in \mathbb{F}_{7^2}$. Let $g=0$. Then the bent function constructed by Lemma 4 is $F(x, y_{1}, y_{2})=f_{Tr_{1}^{2}((z^{13}y_{1}+zy_{2})^{2})}(x)+Tr_{1}^{2}(4z^{12}(y_{1}^{2}+y_{2}^{2})), (x, y_{1}, y_{2})\in \mathbb{F}_{7^2} \times \mathbb{F}_{7^2} \times \mathbb{F}_{7^2}$. It is a self-dual bent function according to 2) of Theorem 5 and it is easy to verify that it is non-quadratic.
\end{example}
\section{Relations between the dual of generalized bent functions and the dual of bent functions}
In this section, we characterize the relations between the generalized bentness of the dual of generalized bent functions and the bentness of the dual of bent functions, as well as the self-duality relations between generalized bent functions and bent functions. The main result is the following theorem:

\begin{theorem}
Let $k\geq 2$ be an integer. Let $f: V_{n}\rightarrow \mathbb{Z}_{p^k}$ be a generalized bent function with the decomposition $f(x)=\sum_{i=0}^{k-1}f_{i}(x)p^{k-1-i}=f_{0}(x)p^{k-1}+\bar{f}_{1}(x)$ where $f_{i}$ is a function from $V_{n}$ to $\mathbb{F}_{p}$ for any $0\leq i \leq k-1$ and $\bar{f}_{1}(x)=\sum_{i=1}^{k-1}f_{i}(x)p^{k-1-i}$ is a function from $V_{n}$ to $\mathbb{Z}_{p^{k-1}}$. For any function $F: \mathbb{F}_{p}^{k-1}\rightarrow \mathbb{F}_{p}$, define $g_{f, F}=f_{0}+F(f_{1}, \dots, f_{k-1})$. Then

1) $f^{*}$ is generalized bent if and only if for any function $F: \mathbb{F}_{p}^{k-1}\rightarrow \mathbb{F}_{p}$, $g_{f, F}^{*}$ is bent.

2) $f$ is self-dual generalized bent if and only if for any function $F: \mathbb{F}_{p}^{k-1}\rightarrow \mathbb{F}_{p}$, $g_{f, F}$ is self-dual bent.
\end{theorem}
\begin{proof}
1) First, by Lemma 1, $g_{f, F}$ is bent for any $F: \mathbb{F}_{p}^{k-1}\rightarrow \mathbb{F}_{p}$. And by Lemma 1, we have $f^{*}(x)=f_{0}^{*}(x)p^{k-1}+\lambda(x)$ and $g_{f, F}^{*}(x)=f_{0}^{*}(x)+F(\lambda_{1}(x), \dots, \lambda_{k-1}(x))$ where $\lambda=\sum_{i=1}^{k-1}\lambda_{i}p^{k-1-i}, \lambda_{i}: V_{n}\rightarrow \mathbb{F}_{p}$ and $\lambda: V_{n}\rightarrow \mathbb{Z}_{p^{k-1}}$ satisfies that for any $a \in V_{n}$, $\sum_{x \in V_{n}: \bar{f}_{1}(x)=\lambda(a)}\zeta_{p}^{f_{0}(x)-\langle a, x\rangle}=W_{f_{0}}(a)$ and $\sum_{x \in V_{n}:\bar{f}_{1}(x)=v}\zeta_{p}^{f_{0}(x)-\langle a, x\rangle}=0$ for any $v \neq \lambda(a)$. Hence, by Lemma 1, $f^{*}$ is generalized bent if and only if for any function $F: \mathbb{F}_{p}^{k-1}\rightarrow \mathbb{F}_{p}$, $g_{f, F}^{*}$ is bent.

2) Suppose $f$ is self-dual generalized bent, that is, $f=f^{*}$. By $f=\sum_{i=0}^{k-1}$ $f_{i}p^{k-1-i}$ and $f^{*}=f_{0}^{*}p^{k-1}+\sum_{i=1}^{k-1}\lambda_{i}p^{k-1-i}$, we have $f_{0}=f_{0}^{*}, f_{i}=\lambda_{i}, 1\leq i \leq k-1$. As $g_{f, F}^{*}=f_{0}^{*}+F(\lambda_{1}, \dots, \lambda_{k-1})$, $g_{f, F}=f_{0}+F(f_{1}, \dots, f_{k-1})=g_{f, F}^{*}$, that is, $g_{f, F}$ is self-dual bent for any function $F: \mathbb{F}_{p}^{k-1}\rightarrow \mathbb{F}_{p}$. Suppose for any function $F: \mathbb{F}_{p}^{k-1}\rightarrow \mathbb{F}_{p}$, $g_{f, F}$ is self-dual bent, that is, $g_{f, F}=g_{f, F}^{*}$ for any function $F: \mathbb{F}_{p}^{k-1}\rightarrow \mathbb{F}_{p}$. Let $F=0$, we obtain $f_{0}=f_{0}^{*}$. Let $F=F_{i}, 1\leq i \leq k-1$ where $F_{i}(x_{1}, \dots, x_{k-1})=x_{i}$, we obtain $f_{i}=\lambda_{i}, 1\leq i \leq k-1$. Hence, $f^{*}=f$, that is, $f$ is self-dual generalized bent.
\end{proof}

By Theorem 6, if $f_{0}^{*}$ is not bent (resp., $f_{0}$ is not self-dual bent), then obviously $f^{*}$ is not generalized bent (resp., $f$ is not self-dual generalized bent). But the inverses are not true. We illustrate this with the following two examples.

\begin{example}
Let $z$ be the primitive element of $\mathbb{F}_{3^5}$ with $z^{5}+2z+1=0$. Let $f: \mathbb{F}_{3^5}\times \mathbb{F}_{3} \times \mathbb{F}_{3} \rightarrow \mathbb{Z}_{3^2}$ be defined as $f=3f_{0}+f_{1}$ where $f_{i}: \mathbb{F}_{3^5}\times \mathbb{F}_{3}\times \mathbb{F}_{3}\rightarrow \mathbb{F}_{3}, i=0,1$, $f_{0}(x, y_{1}, y_{2})=Tr_{1}^{5}(x^2)+(y_{1}+Tr_{1}^{5}(z^{47}x^{2}))(y_{2}+Tr_{1}^{5}(z^{10}x^{2}))$ and $f_{1}(x, y_{1}, y_{2})=y_{2}+Tr_{1}^{5}(z^{10} x^2)$. Then $f$ is a generalized bent function constructed by Theorem 1. One can verify that $f^{*}$ is not generalized bent, but $f_{0}^{*}$ is bent.
\end{example}

\begin{example}
Let $f: \mathbb{F}_{5}^{2}\rightarrow \mathbb{Z}_{5^2}$ be defined as $f=5f_{0}+f_{1}$ where $f_{0}(x_{1}, x_{2})=x_{1}^{2}+x_{2}^{2}, (x_{1}, x_{2})\in \mathbb{F}_{5}^{2}$ and $f_{1}(x_{1}, x_{2})=2x_{1}+x_{2}$. Then $f$ is a generalized bent function and $f^{*}=5g_{0}+g_{1}$ where $g_{0}(x_{1}, x_{2})=x_{1}^{2}+x_{2}^{2}, (x_{1}, x_{2})\in \mathbb{F}_{5}^{2}$ and $g_{1}(x_{1}, x_{2})=x_{1}+3x_{2}, (x_{1}, x_{2})\in \mathbb{F}_{5}^{2}$. Hence, $f$ is not self-dual generalized bent, but $f_{0}$ is self-dual bent.
\end{example}
\section{Conclusion}
In this paper, we study the dual of generalized bent functions $f: V_{n}\rightarrow \mathbb{Z}_{p^k}$ where $V_{n}$ is an $n$-dimensional vector space over $\mathbb{F}_{p}$ and $p$ is an odd prime, $k$ is a positive integer. We give an explicit construction of generalized bent functions whose dual can be generalized bent or not generalized bent. We show that the generalized indirect sum construction method given in \cite{Wang} can provide a secondary construction of generalized bent functions for which the dual can be generalized bent or not generalized bent. For generalized bent functions $f$ whose dual $f^{*}$ is generalized bent, by ideal decomposition in cyclotomic field, we prove $f^{**}(x)=f(-x)$. For generalized bent functions $f$ which satisfy that $f(x)=f(-x)$ and its dual $f^{*}$ is generalized bent, we give a property and as a consequence, we prove that there is no self-dual generalized bent function from $V_{n}$ to $\mathbb{Z}_{p^k}$ if $p\equiv 3 \ (mod \ 4)$ and $n$ is odd. For other cases, we give a secondary construction of self-dual generalized bent functions. In the end, we characterize the relations between the generalized bentness of the dual of generalized bent functions and the bentness of the dual of bent functions, as well as the self-duality relations between generalized bent functions and bent functions.

\end{document}